\documentclass[11pt]{article}
\usepackage[margin=.7in]{geometry}
\usepackage{mymacros}
\usepackage{enumitem}
\usepackage{caption}
\usepackage{float}

\newcommand{\imextended}{\mathsf{IM}\text{-}\mathsf{extended}}
\newcommand{\er}{\mathsf{Er}}
\newcommand{\info}{\textsc{info\_task} }
\newcommand{\imax}{\textup{Imax}}
\newcommand{\good}{\textsc{good}}
\begin{document}
\title{Chain rules for one-shot entropic quantities via operational methods}
\author{Sayantan Chakraborty$^1$ \and Upendra Kapshikar$^1$}
\date{
        ${}^{1}$ Centre for Quantum Technologies, National University of Singapore
            }
            \maketitle
\begin{abstract}
    We introduce a new operational technique for deriving chain rules for general information theoretic quantities.
    This technique is very different from the popular (and in some cases fairly involved) methods like SDP formulation and operator algebra or norm interpolation.
    Instead, our framework considers  a simple information transmission task and obtains lower and upper bounds for it.
    The lower bounds are obtained by leveraging a successive cancellation encoding and decoding technique.
    Pitting the upper and  lower bounds against each other gives us the desired chain rule.
    As a demonstration of this technique, we derive chain rules for the \emph{smooth max mutual information} and the \emph{smooth-Hypothesis testing mutual information}.
\end{abstract}
\section{Introduction}

In 1948, Shannon \cite{Shannon} pioneered the field of information theory by introducing two central problems;  noiseless source coding and noisy channel coding.
To that end, Shannon introduced the notions of Shannon entropy and mutual information, which characterise these two information processing tasks, respectively.
Since then, these two quantities have found numerous applications in many other problems, both within information theory, as well as in cryptography and computer science in general.
For a random variable $X \sim P_X$ , its entropy $H(X)$ is defined as
\[ H(X) = \E\limits_{ x \leftarrow P_X} \left[ \frac{1}{\log(P_{X}(x))} \right].\]
For a joint probability distribution $P_{XY}$, one can analogously define its entropy $H(XY)$;
\[ H(XY) = \E\limits_{ x,y \leftarrow P_{XY}} \left[ \frac{1}{\log(P_{XY}(x,y))} \right].\]
A \emph{chain rule} for the entropy establishes a relationship between the joint entropy and the entropies of the individual variables:
\[
H(XY)= H(X)+H(Y~|~X),
\]
where 
\[
H(Y~|~X)\coloneqq \E\limits_{x \leftarrow  P_X}\left[H(Y~|~X=x)\right]
\]
is the \emph{conditional entropy} of the random variable $Y$ given $X$. 

Such decompositions of joint variable functionalities into individual functionalities are known to hold not only for the entropy function but also for other useful quantities.
For example, consider a tripartite probability distribution $P_{XYZ}$.
Then a chain rule for the mutual information between the systems $XY$ and $Z$ can be written as:
\[
I(XY : Z) = I(X :Z )+ I(Y : Z~|~X).
\]
Chain rules in general are very useful in the design and analysis of information processing protocols, particularly those where multiple parties are present \cite{SlepianWolfMac,Ahlswede:mac,Ahlswede1974,Liao:mac} .
Chain rules for mutual information have been used in contexts other than information-theoretic tasks, for example, in proving direct sum and direct product theorems in communication complexity, \cite{RAZBOROV1992385, JRS_setdisjoint, JRS'03, JRS'05} to name a few (see \cite{Jain} for a more comprehensive list).

The information-theoretic quantities mentioned above can also be defined for more general objects such as quantum states.
For a quantum state $\rho^A$, the \emph{von Neumann entropy}, is defined as
\[
H(A)\coloneqq - \Tr\left[\rho^A \log \rho^A\right].
\]
Analogously, for a bipartite state $\rho^{AB}$, the quantum mutual information is defined as
\[
I(A:B)\coloneqq H(A)+H(B)-H(AB).
\]
However, the conditional entropy of the system $A$ given $B$ cannot be defined in a manner similar to that in the classical case. Thus, in this case, one uses the chain rule itself to define the quantum conditional entropy:
\[
H(A~|~B)\coloneqq H(AB)-H(B).
\]
The chain rule for mutual information follows from the chain rule of $H$ and the definition of $I$.
Furthermore, Jain \cite{Jain} used these chain rules along with the existence of Nash Equilibrium for some suitably defined games to derive a chain rule for the \emph{capacity} of classical-quantum and quantum channels.

The Shannon and von Neumann entropic quantities although useful in characterizing many important information-processing tasks, are somewhat restricted. 
They are most useful in settings where many \emph{independent} copies of the underlying resource are available. 
For example, in quantum source compression one exhibits an algorithm to compress the quantum state $\rho^{\otimes n}$ using only $nH(A)$ many qubits.
For this, it is usually assumed that $n$ copies of a quantum state $\rho$ are available, in order to show that there exists a compression algorithm.
Comparatively, a  more natural framework is that of \emph{one-shot} information theory which  considers the setting where only \emph{one} copy of the underlying resource is available.
There exists a rich body of work that explores information theoretic questions in this setting with the aid of the smooth min and max entropy formalism. This formalism was introduced and developed by a series of papers \cite{Renner_thesis, Renner_Wolf,Konig_Renner_1, Datta_entanglement_monotone,  Konig_Renner_Schaffner,Duality_smooth_minmax, QAEP, Renes_Renner} in the context of both information-theoretic and cryptographic applications.
The (conditional) smoothed min and max entropies ($H_{\max}^{\eps}(A~|~B)$ and $H_{\min}^{\eps}(A~|~B)$, respectively) are \emph{robust} versions of the corresponding unsmoothed quantities.
Here the parameter $\eps$, referred to as the smoothing parameter, is used to specify the accuracy of certain protocols.
For example, the smooth min-entropy $H_{\min}^{\eps}(A~|~B)$ characterizes the number of (almost) random bits one can extract from the system $A$ when an adversary is in possession of the system $B$.
The parameter $\eps$ here denotes the requirement that the random bits produced in such an extraction should have a bias of at most $\eps$ (see \cite{Dupuis_Wullscheger_Renner,Dupuis_thesis}).
Similarly, the quantity $H_{\max}^{\eps}(A~|~B)$ characterizes the number of entangled qubits required for \emph{state merging} \cite{State_merging, Berta_thesis}.
Thus, given their importance, a natural question is whether these quantities obey chain rules similar to their von Neumann counterparts. This question was investigated in the work of Vitanov et al. \cite{Dupuis_smooth_min_max}, where the authors provided several chain rules for the smooth min max entropies.
It is worth pointing out that the chain rules that one gets for such quantities are only one-sided chain rules, in that they are inequality expressions rather than equality.
\begin{example}
In \cite{Dupuis_smooth_min_max}, Vitanov et al. showed the following chain rule  for the smooth min-entropy (ignoring additive log terms): Given a quantum state $\rho^{ABC}$ and $\eps, \eps^{\prime}, \eps^{\prime\prime}>0$ such that $\eps> \eps^{\prime}+2\eps^{\prime\prime}$, it holds that:
\[
H_{\min}^{\eps}(AB \vert C) \geq H_{\min}^{\eps^{\prime}}(A \vert C)+H_{\min}^{\eps^{\prime\prime}}(A \vert BC).
\]
\end{example}
Dupuis further showed similar chain rules  for the sandwiched R\'enyi $\alpha$-entropies in \cite{Dupuis_alpha_entropy}.

Although the smoothed min and max entropy formalism has proven to be very useful in the description of several quantum information processing tasks, it does not tell the whole story.
The works of Anshu et al. \cite{Convex_Split, Anshu_optimality, quantum_assistated_classical_anurag} and  Wang and Renner \cite{Wang_Renner} highlight the importance of smooth max divergence, the smooth hypothesis testing divergence and their derivative quantities  (see Section \ref{sec:def} for the relevant definitions).
Wang and Renner characterised the one-shot capacity of a classical-quantum channel $\mathcal{N}^{X\to B}$ in terms of the smooth hypothesis testing mutual information:
\[
\max \limits_{P_X} I_H^{\eps}(X:B).
\]
A similar characterisation for the entanglement-assisted classical capacity of a channel $\mathcal{N}^{A\to B}$  was given in the work of Anshu et al. \cite{quantum_assistated_classical_anurag}, who showed the assisted classical capacity of any quantum channel is given by
\[
\max\limits_{\ket{\varphi}^{RA}}I_H^{\eps}(R:B)_{\I^R\otimes \mathcal{N}^{A\to B}\left(\varphi^{RA}\right)}.
\]
Another important quantity \emph{smooth max mutual information} $I_{\max}^{\eps}(A:B)$ gives an achievable quantum communication cost for the state redistribution problem \cite{Convex_Split} and  state splitting \cite{Convex_Split,Berta_Reverse_Shannon}. Unlike their smooth max min entropic counterparts, to the best of our knowledge, the existence of chain rules for these important information quantities has not received much attention.
Our goal in this paper is to introduce techniques that will enable us to present chain rules for these quantities.
\subsection{Our Contribution}
The main results that we present in this work are as follows:
\begin{theorem}{\bf [Informal]}\label{thm:mainthmintro} {For any $\varepsilon>0$} and
any quantum state $\rho^{ABC}$,
it holds that
\[
I_H^{\eps}(AB:C) \geq I_H^{\eps^{\prime}}(A:C)+I_H^{\eps^{\prime\prime}}(B:AC) -I_{\max}(A:B) +O\left(\log \eps\right)
\]
where, $\eps^{\prime}$ and $\eps^{\prime\prime}$ are $O(\eps^2)$.
\end{theorem}

\begin{theorem}{\bf [Informal]}\label{thm:mainthmintro2}
{For any $\varepsilon>0$} and any quantum state $\rho^{ABC}$,
it holds that
\[
I_{\max}^{\eps}(AB:C) \leq I_{\max}^{\eps^{\prime}}(A:C)+I_{\max}^{\eps^{\prime\prime}}(B:AC) -I_{H}^{\eps^{\prime\prime\prime}}(A:B) -O\left(\log \eps\right)
\]
where, $\eps^{\prime}$, $\eps^{\prime\prime}$ and  $\eps^{\prime\prime\prime}$ are $O(\eps^2)$.
\end{theorem}

\begin{remark}
We should mention that it is not at all obvious how to prove \cref{thm:mainthmintro} using standard techniques in one-shot information theory.
One can suspect that due to a close connection between the smooth hypothesis testing divergence and the information spectrum divergence, it might be possible to arrive at a chain rule like Theorem~\ref{thm:mainthmintro}.
Indeed, exploiting the said relation one can prove the following statement (ignoring additive log factors):
\[
I_{H}^{\epsilon} \left(AB : C\right)_{\rho} \geq I_{H}^{\epsilon} \left(A : C\right)_{\rho}  + D_{H}^{\epsilon} \left( \rho^{ABC} ~\Vert~ \rho^{B} \otimes \Pi_{s}^{AC}~ \rho^{AC}~\Pi_{s}^{AC} \right) \]
where $\Pi^{AC}_s$ is the information spectrum projector. However, it is not clear how to remove this projector from the expression above to get the desired chain rules, since in general, it does not commute with $\rho^{AC}$.
\end{remark}

\begin{remark}
In \cite{Dupuis_HH}, Dupuis et al. showed chain rules for the smooth hypothesis testing conditional entropy $H_H^{\eps}$ using a chain rule for the smooth hypothesis testing divergence between an arbitrary state $\rho$ and a state $\sigma$ which is invariant under some group action. However, it is not clear how this technique can be used to prove the chain rule claimed in \cref{thm:mainthmintro}. 
\end{remark}

\begin{remark}
Chakraborty et al. proved a weaker version of \cref{thm:mainthmintro2} in \cite{CNS}. In particular, the authors in that paper proved the following bound:
\[
I_{\max}^{\eps}(AB:C) \leq I_{\max}^{\eps^{\prime}}(A:C)+I_{\max}^{\eps^{\prime\prime}}(B:AC)  -O\left(\log \eps\right).
\]
 We present a sharper version of this inequality in this paper.
\end{remark}
\subsection*{Organisation of the paper}
The paper is organised as follows: 
In Section~\ref{sec:def} we present relevant definitions and facts that will be useful throughout the paper.
In \cref{sec:overview1} we present an overview of the main operational method that we use to prove Theorems \ref{thm:mainthmintro} and \ref{thm:mainthmintro2}. In this section, we also show how an application of these ideas leads directly to the proof of \cref{thm:mainthmintro2}. In \cref{sec:overview2} we explain why the ideas presented in \cref{sec:overview1} cannot be directly applied to prove \cref{thm:mainthmintro}. In this section, we also present a weaker version of \cref{thm:mainthmintro}, called \cref{prop:IMstatechainrule}, which is a result akin to \cref{thm:mainthmintro} but valid only for a specific subclass of quantum states, which we call $\mathsf{IM}$-states (see \cref{sec:overviewImstates}). We introduce this proposition for the sake of demonstrating the main ideas that eventually go into the proof of \cref{thm:mainthmintro}. In \cref{sec:overviewImstates} we present an overview of our proof for \cref{prop:IMstatechainrule}, followed by the formal definitions and techniques in Sections \ref{sec:IMstatesandRedcution} and \ref{sec:ChainRules}. Finally, in \cref{sec:overviewfullthm}, we present the full proof of \cref{thm:mainthmintro}. 
\section{Preliminaries}
 \subsection{Definitions}\label{sec:def}
\begin{definition}{\bf (Smooth Hypothesis Testing Relative Entropy)}
The smooth min-relative entropy $D_{H}^{\epsilon}$ between two states $\rho$ and $\sigma$ is defined via the equation below:
\[
2^{-D_H^{\eps}(\rho~||~\sigma)}\coloneqq \min\limits_{\substack{0\leq \Pi\leq \I\\ \Tr\left[\Pi~\rho\right]\geq 1-\eps}} \Tr \left(\Pi~ \sigma\right)
\]
\end{definition}
\noindent Using the usual correspondence between entropy and mutual information, one can define smooth Hypothesis testing mutual information in a state $\rho$;
\[ I_{H}^{\epsilon}(A:B)_{\rho}= D^{\epsilon}_{H}\left(\rho^{AB} \Vert \rho^A \otimes \rho^B\right).\]
Given the context of our work, we will be mostly interested in smooth Hypothesis testing mutual information of a particular state associated with a channel.

\begin{definition}[{\bf An optimal tester for $(I, \epsilon, \rho, \mathcal{N})$}] \label{def:optimal_projector}
Let $\mathcal{N}^{A \rightarrow B}$ be a channel and $\rho^{AC}$ be a pure state.
Then,
\[
\begin{aligned}
I^{\epsilon}_{H}(B:C)_{\mathcal{N}(\rho^{AC})}  &= D^{\epsilon}_{H}\left( \mathcal{N}\left(\rho^{AC}\right) \Vert \mathcal{N}\left( \rho^A \right) \otimes \rho^C\right) \\ &= - \log~\min\limits_{\substack{0\leq \Pi\leq \I^{BC}\\ \Tr\left[\Pi \left(\mathcal{N}(\rho^{AC})\right)\right]\geq 1-\eps}} \Tr \left[\Pi~ \left(\mathcal{N}\left( \rho^A \right) \otimes \rho^C \right)\right].
\end{aligned}
\]
An operator $\Pi$ that achieves the optimum in the above equation will be referred to as an optimal tester for $(I, \epsilon, \rho, \mathcal{N})$.
\end{definition}
\noindent Thus, it follows from the definition that, if $\Pi$ is an optimal tester for $(I, \epsilon, \rho, \mathcal{N})$ then, 
\begin{align}
    2^{- I^{\epsilon}_{H}(B:C)_{\mathcal{N}(\rho^{AC})}} &=  \Tr \left[\Pi\left(\mathcal{N}\left( \rho^A \right) \otimes \rho^C \right)\right] \\
    \Tr\left[\Pi\left(\mathcal{N}(\rho^{AC})\right)\right]&\geq 1-\eps.
\end{align}
\begin{definition}[{\bf Max Relative Entropy}]
Given quantum states $\rho$ and $\sigma$ such that $\textup{supp}(\rho)\subseteq \textup{supp}(\sigma)$, the max relative entropy $D_{\max}(\rho~||~\sigma)$ is defined as
\[
D_{\max}\coloneqq \inf\brak{\lambda~|~\rho\leq 2^{\lambda}\sigma}.
\]
\end{definition}
\noindent Again, using the usual correspondence between entropy and mutual information, one can define the max mutual information with respect to a state $\rho^{AB}$ as:
\[
I_{\max}(A:B)_{\rho^{AB}}\coloneqq D_{\max}(\rho^{AB}~||~\rho^A\otimes \rho^B)
\]
\begin{definition}[{\bf Smooth Max Relative Entropy}]
Given quantum states $\rho$ and $\sigma$ such that $\textup{supp}(\rho)\subseteq \textup{supp}(\sigma)$, let $\mathcal{B}^{\eps}(\rho)$ be the $\eps$ ball around the state $\rho$;
\[
\mathcal{B}^{\eps}(\rho)\coloneqq \brak{\tau\geq 0~|~\norm{\tau-\rho}\leq \eps, \Tr\left[\tau\right]\leq 1}.
\]
Then the smooth max relative entropy $D_{\max}^{\eps}(\rho~||~\sigma)$ is defined as
\[
D_{\max}^{\eps}(\rho~||~\sigma)\coloneqq \inf\limits_{\rho^{\prime}\in \mathcal{B}^{\eps}(\rho)}D_{\max}(\rho^{\prime}~||~\sigma).
\]
\end{definition}
\noindent Similarly, the smooth max mutual information with respect to a state $\rho^{AB}$ is defined as:
\[
I_{\max}^{\eps}(A:B)_{\rho^{AB}}\coloneqq \inf\limits_{\rho^{\prime AB}\in \mathcal{B}^{\eps}(\rho^{AB})}I_{\max}(A:B)_{\rho^{{\prime}^ {AB}}}.
\]
\subsection{Facts}
\begin{fact}[{\bf Gentle Measurement Lemma}]
\label{fact:gentle_measurement}
    Let $\rho$ be a state and $\lbrace\Lambda_i\rbrace_i$ be a POVM such that there exists an $i_0$ with \[\Tr\left( \Pi_{i_0} \rho\right) \geq 1- \epsilon.\]
    Let \[\rho^\prime = \sum\limits_{i}{\sqrt{\Lambda_i} \rho \sqrt{\Lambda_i}} \otimes \ketbra{i}\] be the post measurement state.
    Then, \[\Vert \rho \otimes \ketbra{i_0}-\rho^\prime \Vert_1 \leq 3 \sqrt{\epsilon}.\]
\end{fact}

\begin{fact}[{\bf Uhlmann's Theorem~\cite{uhlmann76}}]
\label{uhlmann_exact}
Let $\rho^A \in \mathcal{D}(\mathcal{H}_A)$ be a state and let $\rho^{AB}\in \mathcal{D}(\mathcal{H}_{AB}),\ \rho^{AC} \in \mathcal{D}\left( \mathcal{H}_{AC}\right)$ be purifications of $\rho_A$.
Then there exists an isometry $V^{C \rightarrow B}$ \emph{(}from a subspace of $\mathcal{H}_C$ to a subspace of $\mathcal{H}_B$\emph{)} such that,
\[  \mathbb{I}_A \otimes V^{C \rightarrow B}  \left( \rho^{AC}\right) = \rho^{AB}.\]
\end{fact}
\begin{fact}[{\bf Closeness~\cite[Fact 9]{Wang_Renner,Anshu_optimality}}]\label{fact:closeness}
Let $\phi^{MM'}$ be a quantum state that satisfies the following conditions:
\begin{align*}
    &\phi^{M}=\frac{\I}{\abs{M}} \hspace{1.23cm} \text{and}
    &\Tr\left[\sum_m \ketbra{m}^M\otimes \ketbra{m}^{M'}\phi^{MM'}\right]\geq 1-\eps.
\end{align*}
Then for any quantum state $\sigma^{M'}$, it holds that 
\[
D_H^{\eps}(\phi^{MM'}~||~\phi^M\otimes \sigma^{M'})\geq \log \abs{M} .
\]
\end{fact}

\section{Overview of Techniques}
In this section, we present the main ideas that lead to the proofs of Theorems \ref{thm:mainthmintro} and \ref{thm:mainthmintro2}.
\subsection{The Main Idea}\label{sec:overview1}
The main techniques used thus far to prove chain rules for the smooth min and max entropies \cite{Dupuis_smooth_min_max} and the R\'enyi $\alpha$-entropies have involved the SDP formulations of these quantities, or norm interpolation methods. While these techniques are extremely sophisticated and powerful, in this paper we take a much simpler operational approach. 
The following observation is at the heart of our approach:
\vspace{1mm} \\
Consider a situation where two parties, Alice and Bob wish to perform a generic information processing task \info using a resource state $\rho^{AB}$ and communication.
Suppose we are promised the following:
\begin{enumerate}
    \item \emph{Any} protocol which achieves \info  using the resource state $\rho^{AB}$ requires Alice and Bob to communicate at least $C(A\to B)_{\rho^{AB}}$ number of bits, where the $C(\cdot)$ is a function of the state $\rho^{AB}$.
    \item There exists a protocol $\mathcal{P}(A\to B)$ which achieves \info~  using the state $\rho^{AB}$, with a communication cost $C(A\to B)_{\rho^{AB}}$.
    \item Additionally, $\mathcal{P}(A\to B)$ ensures that at the end of the protocol, the share $A$ of the state $\rho^{AB}$ belongs to Bob and both the classical and quantum correlations between the $A$ and $B$ remain intact.
\end{enumerate}
One can consider a successive cancellation strategy for achieving \info :
Consider a situation where Alice and Bob wish to achieve \info with the resource state $\rho^{A_1A_2B}$. Consider the following strategy:
\begin{enumerate}
    \item Alice enacts the protocol $\mathcal{P}(A_1\to B)$ using only the marginal $\rho^{A_1B}$, at a communication cost of $C(A_1\to B)_{\rho^{A_1B}}$. At the end of $\mathcal{P}(A_1\to B)$, the $A_1$ share of the resource state $\rho^{A_1A_2B}$ belongs to Bob.
    \item Alice can then enact the protocol $\mathcal{P}(A_2\to A_1B)$, with a communication cost $C(A_2\to A_1B)_{\rho^{A_1A_2B}}$.
\end{enumerate}
The above protocol achieves \info while using the resource state $\rho^{A_1A_2B}$ with a cumulative
\[
C(A_1\to B)_{\rho^{A_1B}}+C(A_2\to A_1B)_{\rho^{A_1A_2B}}
\]
of communication. Then, using the promised lower bound on the amount of communication required to achieve this task, we see that
\[
\begin{aligned}
&C(A_1\to B)_{\rho^{A_1B}}+C(A_2\to A_1B)_{\rho^{A_1A_2B}} \geq C(A_1A_2\to B)_{\rho^{A_1A_2B}}.
\end{aligned}
\]
This algorithmic technique of showing the existence of chain rules was first exploited by Chakraborty et al.\cite{CNS} to demonstrate chain rules for the smooth max mutual information.
As mentioned previously, we present below an improved version of this result in \cref{thm:mainthmintro2}.
The proof idea is as follows: \\ \vspace{1mm} \\ Consider the task of \emph{quantum state splitting}, in which a party (say Alice) holds the $AM$ share of the pure quantum state $\ket{\varphi}^{RAM}$ at the beginning of the protocol, $R$ being held by the referee. Alice is then required to send the $M$ portion of the state to Bob, while trying to minimize the number of qubits communicated to Bob. It is known \cite{Berta_Reverse_Shannon} that for this problem, Alice needs to communicate at least 
\[
\frac{1}{2} I_{\max}^{\eps}(R:M)
\]
number of qubits to Bob. To show the chain rule, consider a pure stage $\ket{\varphi}^{RAM_1M_2}$. Then:
\begin{enumerate}
    \item First Alice sends the $M_1$ system to Bob using state redistribution protocol \cite{Convex_Split}. At this point, the global state is some $\ket{\varphi^{\prime}}^{RAM_1M_2}$, which is $\eps$ close (in the purified distance) to the original state $\ket{\varphi}^{RAM_1M_2}$, with the $M_1$ system being in the possession of Bob.
    \item To do this, Alice communicated $\frac{1}{2} I_{\max}^{\eps}(R:M_1)$ qubits to Bob (suppressing the additive log terms).
    \item Next, Alice sends the system $M_2$ to Bob. Note that, instead of using the vanilla state redistribution protocol once more (which would cost about $\frac{1}{2}I_{\max}^{\eps}(RM_1:M_2)$ qubits of communication), we can take advantage of the fact that Bob possesses some side information about the state, in particular, the register $M_1$ already in his possession. Anshu et al. presented a modified state redistribution protocol in \cite{AJW_modified} which does precisely this, while reducing the quantum communication cost to 
    \[
    \frac{1}{2}\cdot \left(I_{\max}^{\eps}(RM_1:M_2)-I_H^{\eps}(M_1:M_2)\right).
    \]
\end{enumerate}
Putting the achievable communication rate derived above against the lower bound shown by \cite{Berta_Reverse_Shannon} then gives us the chain rule:
\[
I_{\max}^{\eps^{\prime}}(R:M_1M_2) \leq I_{\max}^{\eps}(R:M_1) + I_{\max}^{\eps}(RM_1:M_2)-I_H^{\eps}(M_1:M_2)
\]
where we have ignored the additive log terms and set $\eps^{\prime}$ to reflect the total error made by the achievable strategy. The explicit computation of the error is easy and follows along similar lines to the calculation presented in \cite{CNS} with some minor tweaks. Hence we do not repeat it here. Instead, the rest of the paper is devoted to proving \cref{thm:mainthmintro}, which is technically much harder to prove.
\subsection{Issues with $I_H^{\eps}$}\label{sec:overview2}
The idea presented in \cref{sec:overview1} can similarly be used to prove chain rules where the direction of the inequality is reversed. In that case, one has to consider a task for which Alice and Bob wish to \emph{maximize} the amount of communication, and there exists a known upper bound. However, this idea cannot be readily applied when trying to prove chain rules for $I_H^{\eps}$. The difficulties are as follows: \\ \vspace{1mm} \\
Suppose we wish to derive a chain rule of the form
\[
C(A_1A_2\to B)_{\rho^{A_1A_2B}} \geq C(A_1\to B)+ C(A_2\to A_1B)
\]
we require the existence of an information processing task for which the \emph{maximum} number of bits that can be transmitted is quantified by $C(A_1A_2\to B)$. Note that this number is a function of a specific fixed state $\rho^{A_1A_2B}$. For $I_H^{\eps}$, a natural task that one may consider for this purpose is entanglement-assisted channel coding over some quantum channel $\mathcal{N}^{A\to B}$. However, as mentioned before, Anshu et al. showed in \cite{quantum_assistated_classical_anurag,Anshu_optimality} that the maximum number of bits that can be sent using this channel is given by
\[
\max\limits_{\ket{\varphi}^{A_1A_2A}}I_H^{\eps}(A_1A_2:B)_{\I^{A_1A_2}\otimes \mathcal{N}^{A\to B}\left(\varphi^{A_1A_2A}\right)}.
\]
Note that the above capacity expression is a function of the \emph{channel} and \emph{not} a fixed state.
In particular, there is a maximization over state $\ket{\phi}$.
This prevents us from directly importing our operational approach here. Note that this was not an issue in the case of $I_{\max}^{\eps}$ since the task of state splitting is defined for a specific fixed state, and not a channel, and neither did it involve any maximization.
To remedy this situation, we need to do the following:
\begin{enumerate}
    \item \label{requirement:channel} Given a state $\rho^{A_1A_2B}$ we need to exhibit a channel $\mathcal{N}^{A\to B}$ and pure state $\ket{\varphi}^{A_1A_2A}$ such that
    \[
    \mathcal{N}^{A\to B}\left(\varphi^{A_1A_2A}\right)= \rho^{A_1A_2B}.
    \]
    \item \label{requirement:protocol} Having exhibited this channel, we need to show that \emph{any} protocol which uses $\ket{\varphi}^{A_1A_2A}$ as a shared entangled state and sending classical messages across $\mathcal{N}$,  can send at most $I_H^{\eps}(A_1A_2:B)$ many bits (and error at most $\epsilon$).
\end{enumerate}
We refer to above two conditions as \textbf{Requirement}~\ref{requirement:channel} and \textbf{Requirement}~\ref{requirement:protocol}, respectively.
Before going to chain rules for an arbitrary state, we first show the following preposition, which includes some core ideas of our protocol.

\begin{proposition}\label{prop:IMstatechainrule}
Given a quantum state $\rho^{ABC}$ such that
\[
\Tr_C\left[\rho^{ABC}\right]=\rho^A\otimes \rho^B,
\]
it holds that
\[
I_H^{\eps}(AB:C) \geq I_H^{\eps^{\prime}}(A:C)+I_H^{\eps^{\prime\prime}}(B:AC)+O\left(\log \eps\right)
\]
where both $\eps^{\prime}$ and $\eps^{\prime\prime}$ are $O(\eps^2)$.
\end{proposition}
\noindent Note that when $A$ and $B$ marginals are in tensor, $I_{\max}(A:B)_{\rho}=0$, and hence the above Preposition  exactly recovers the chain rule we wanted.
Throughout the paper we call such states (where marginals are independent) as $\mathsf{IM}$-states.
\subsection{A Warm-up: Chain Rules for $\mathsf{IM}$-States}\label{sec:overviewImstates}
In this section, we introduce the proof techniques that we will use to prove \cref{prop:IMstatechainrule}. See \cref{sec:ChainRules} for a complete proof. We begin by defining a certain subfamily of tripartite states.
The family that we will be interested in, will be such that its marginals on one of the pairs will be independent.
\[\mathcal{F}_{\mathsf{IM}}^{A_fB_fC} = \lbrace\rho^{A_fB_fC} \in \mathcal{D}\left( \mathcal{H}_{A_fB_fC}\right)\ \vert\ \rho^{A_fB_f} = \rho^{A_f} \otimes \rho^{B_f} \rbrace. \]
To mean that $\rho^{A_fB_fC} \in \mathcal{F}_{\mathsf{IM}}^{A_fB_fC}$, we will use the shorthand $\rho^{A_fB_fC}$ is an $(A_fB_fC)-\mathsf{IM}$ state.
Note that the order of register $A_fB_fC$ matters as the marginals only on $A_f$ and $B_f$ are independent. Recall that to prove chain rules for $I_H^{\eps}$ we need to fulfill the Requirements \ref{requirement:channel} and \ref{requirement:protocol}. Requirement \ref{requirement:channel} is not hard to satisfy for $\mathsf{IM}$-states, as is shown by the following lemma:

\begin{lemma}\label{lem:SmChannel} 
For every $(A_fB_fC)-\mathsf{IM}$ state $\rho^{A_fB_fC}$ and purifications $\varphi^{A_fA}_1, \varphi^{B_fB}_2$ of $\rho^{A_f}, \rho^{B_f}$ respectively, there exists a channel $\mathcal{N}^{AB\to C}$ such that the following holds:
\[
\I^{A_fB_f}\otimes \mathcal{N}^{AB\to C}\left(\varphi^{A_fA}_1\otimes \varphi^{B_fB}_2\right)= \rho^{A_fB_fC}.
\]
\end{lemma}

\begin{proof}
Consider a purification ${\rho}^{A_fB_fCR}$ of $\rho^{A_fB_fC}$. Note that this is also a valid purification of the state ${\rho^{A_fB_f}=} \rho^{A_f}\otimes \rho^{B_f}$. Then, by the Uhlmann's Theorem (Fact~\ref{uhlmann_exact})there exists an isometry $V^{AB\to CR}$ such that 
\[
\mathbb{I}^{A_fB_f} \otimes V^{AB\to CR} \left(\varphi^{A_fA}_1\otimes \varphi^{B_fB}_2\right) = {\rho}^{A_fB_fCR}.
\]

Define
\[
\mathcal{N}^{AB\to C}\coloneqq \Tr_R \circ V^{AB\to CR}.
\]
Then it is easy to see that
\[
\I^{A_fB_f}\otimes \mathcal{N}^{AB\to C}\left(\varphi^{A_fA}_1\otimes \varphi^{B_fB}_2\right)= \rho^{A_fB_fC}.
\]
This concludes the proof.
\end{proof}

\begin{remark}
Given an $(A_fB_fC)-\mathsf{IM}$ state $\rho^{A_fB_fC}$, note that the channel which satisfies the conditions of  \cref{lem:SmChannel} is not unique, but instead depends on the purifications $\ket{\varphi}^{A_fA}_1$ and $\ket{\varphi}^{B_fB}_2$. Nevertheless, we will always fix these purifications, and refer to \emph{the} channel constructed in  \cref{lem:SmChannel} as the $\imextended$ channel of $\left(\rho^{A_fB_fC}, \varphi^{A_fA}_1, \varphi^{B_fB}_2 \right)$. 
When the registers are clear from the context, we will denote it as the $\imextended$ channel of $(\rho,\varphi_1,\varphi_2)$.
\end{remark}
\noindent Requirement \ref{requirement:protocol} is much harder to show. To prove that this requirement indeed holds, we use the following idea:
\begin{enumerate}
    \item We first consider the set of \emph{all} entanglement assisted protocols which use the channel $\mathcal{N}^{AB\to C}$ to send classical messages, with an error at most $\eps$. We call this set $\mathcal{S}(\mathcal{N}, \eps)$.
    \item We divide this set into disjoint subsets $\mathcal{S}^{\sigma}(\mathcal{N},\eps)$, where each subset consists of all those protocols whose encoders create some \emph{fixed} state $\sigma^{AB}$ on the system which is input to the channel when all other systems are traced out.
    \item  \label{point:3} We then show that, for a fixed $\sigma^{AB}$, \emph{any} protocol in the set $\mathcal{S}^{\sigma^{AB}}(\mathcal{N},\eps)$ can send at most 
    \[
    I_{H}^{\eps}(A_fB_f: C)_{\I^{A_fB_f}\otimes \mathcal{N}(\ket{\sigma}^{ABA_fB_f})}
    \]
    number of bits through $\mathcal{N}$, where $\ket{\sigma}^{ABA_fB_f}$ is an arbitrary purification of $\sigma^{AB}$. We do this by using a slightly modified form of the converse shown by Anshu et al. \cite{quantum_assistated_classical_anurag}.
    \item For the case of $\mathsf{IM}$-states, setting 
    \[
    \sigma^{AB}\gets \varphi_1^{A_f}\otimes \varphi_2^{B_f}
    \]
    and
    \[
    \ket{\sigma}^{ABA_fB_f}\gets \ket{\varphi_1}^{AA_f}\ket{\varphi_2}^{BB_f}
    \]
    completes the argument.
\end{enumerate}
We explore the above idea of partitioning the set of all protocols with a small error in  \cref{sec:IMstatesandRedcution}. The precise definition of the terms $\mathcal{S}(\mathcal{N},\eps)$ and $\mathcal{S}^{\sigma}(\mathcal{N},\eps)$ can be found in \cref{sec:setup}. The proof of the upper bound for a fixed partition referred to in Point \ref{point:3} can be found in \cref{sec:partition}.

To complete the proof of the chain rule we still need to show a successive coding strategy using the states $\ket{\varphi_1}^{AA_f}\ket{\varphi_2}^{BB_f}$ as a shared resource. To do this we use a standard successive cancellation style argument using Anshu et al.'s coding strategy for entanglement-assisted classical message transmission \cite{quantum_assistated_classical_anurag}. Details can be found in \cref{sec:ChainRules} and Appendix \ref{sec:AJW}.

\section{Partitioning the Space of Good Protocols}\label{sec:IMstatesandRedcution}

In this section, we introduce the partitioning idea, referred to in \cref{sec:overviewImstates} that is key to the proof of our chain rule. Towards that end, we clarify the definition of an entanglement-assisted classical communication protocol (over a noisy quantum channel) in \cref{sec:setup}, and go on to define the set of all such protocols which make a small amount of error. Then in \cref{sec:partition} we introduce a way to partition this set and provide upper bounds on the rates of communication of the protocols that belong to a fixed partition.

\subsection{The Setup for Entanglement Assisted Classical Communication}\label{sec:setup}

\begin{figure}[hbtp]
  \caption{Setup for Entanglement assisted classical communication}
    \label{EA_Classical_Setup}
\fbox{\parbox{\textwidth}{
Consider an entanglement-assisted classical message transmission protocol over a channel  $\mathcal{N}^{A\to B}$, which makes an error at most $\eps$. 
Any such protocol consists of the following objects:
\begin{enumerate}
    \item A state 
    \[
    \psi^{MM_A}\coloneqq\sum\limits_{m\in M} \frac{1}{\vert M \vert }\ketbra{m}^M\otimes \ketbra{m}^{M_A}
    \]
    held by the sender Alice.
    \item Shared entanglement modeled by a pure state
    \[
    \ketbra{\varphi}^{E_AE_B}.
    \]
    where the $E_A$ system is held by Alice and the $E_B$ system is held by the receiver Bob.
    \item An encoder $\mathcal{E}^{M_AE_A\to A}$ which takes as input the states in the  $M_A$  and $E_A$ systems and produces a state on the register $A$, which is the input to the channel $\mathcal{N}^{A \rightarrow B}$.
    \item A decoder $\mathcal{D}^{BE_B\to \widehat{M}}$, which acts on the register $B$ (the output  of the channel), as well as Bob's share of the entanglement, and produces a classical register $\widehat{M}$ which contains a guess for the message sent by Alice.
\end{enumerate}
The protocol is said to make an {(average)}-error at most $\eps$ if
\[\label{eq:cond}
\norm{\mathcal{D}\circ \mathcal{N}\circ \mathcal{E}~\left(\psi^{MM_A}\otimes \varphi^{E_AE_B}\right)-\psi^{MM_A}}_1 \leq \eps \tag{1}
\]
}}
\end{figure}
\begin{definition}
\begin{enumerate}[label=\Alph*.]
\item A protocol $\mathcal{P}$ will be labelled by a tuple $\left( M, {\mathcal{N}}, \mathcal{E},  \mathcal{D}, \ket{\varphi}^{E_AE_B}\right)$.
The average error of the protocol $\er$ is given by the following expression:
\[\begin{aligned}
\er(\mathcal{P}) &= \er\left( M, {\mathcal{N}}, \mathcal{E},  \mathcal{D}, \ket{\varphi}^{E_AE_B}\right) \\ &= \norm{\mathcal{D}\circ \mathcal{N}\circ \mathcal{E}~\left(\psi^{MM_A}\otimes \varphi^{E_AE_B}\right)-\psi^{MM_A}}_1.\end{aligned} \] 
    \item Let $\mathcal{S}(\mathcal{N},\eps)$ to be the set of all protocols $\mathcal{P}$ which makes an error at most $\eps$ while using channel the $\mathcal{N}$.{
    \[
    \mathcal{S}(\mathcal{N},\eps)  = \Big\lbrace \mathcal{P}: \exists M,\mathcal{E},\mathcal{D} \text{ such that } \mathcal{P} \text{ is an } \left( M, {\mathcal{N}},  \mathcal{E},  \mathcal{D}, \ket{\varphi}^{E_AE_B}\right) \text{ and  } \er\left( M, {\mathcal{N}}, \mathcal{E},  \mathcal{D}, \ket{\varphi}^{E_AE_B}\right) \leq \epsilon\Big\rbrace.\]}        
    \item We define $\mathcal{S}^{\rho^A}\left( \mathcal{N}, \epsilon\right) \subseteq \mathcal{S}(\mathcal{N},\eps)$ to be the set of all those protocols $\mathcal{P}\in \mathcal{S}(\mathcal{N},\eps)$ for which the state at the input to the channel is $\rho^A$.
      { \[ \Tr_{ME_B} \left[\mathcal{E}~\left(\psi^{MM_A}\otimes \varphi^{E_AE_B}\right) \right] = \rho^A.\]}
\end{enumerate}

\end{definition}
\subsection{The Partitions and Corresponding Upper Bounds}\label{sec:partition}
Note that $\mathcal{S}^{\rho^A}(\mathcal{N}, \epsilon)$ partitions the set $\mathcal{S}(\mathcal{N}, \epsilon)$. That is, given $\mathcal{P} \in \mathcal{S}(\mathcal{N}, \epsilon)$, there exists a unique $\rho^A$ such that $\mathcal{P} \in \mathcal{S}^{\rho^A}(\mathcal{N}, \epsilon)$.
In~\cite{quantum_assistated_classical_anurag} Anshu, Jain and Warsi showed that for any protocol $\mathcal{P} \in \mathcal{S}(\mathcal{N},\epsilon)$ the number of messages $M$ can be upper bounded by $D_H^{\epsilon}$ of certain states associated with the protocol. 
For our proof, we need a slightly finer version of an analogous statement. 
In the following lemma, we note that such a statement remains valid over individual partitions as well.
 Our proof follows a similar strategy as theirs. 
 We include it here for the sake of completeness.
\begin{lemma}\label{lem:converse}
Let $\mathcal{P}$ be an arbitrary $\left(M,{\mathcal{N}}, \mathcal{E}, \mathcal{D}, \ket{\varphi}^{E_AE_B}\right)$ protocol in $\mathcal{S}^{\rho^A}(\mathcal{N},\eps)$.
Then,
\[
\log \vert M \vert \leq \min_{\sigma^B} D_H^{\eps}\left(\mathcal{N}(\tau^{AB'})~||~\sigma^B \otimes \tau^{B'}\right),
\]
where $\ket{\tau}^{AB^\prime}$ is an {arbitrary} purification of the state $\rho^A$.
\end{lemma}
\begin{proof}
See Appendix \ref{appendix:proof of lem:converse}
\end{proof}

The following corollary follows immediately by setting $\sigma^B=\mathcal{N}(\rho^A)$:

\begin{corollary}\label{corol:usableconverse}
Given the setting of \cref{lem:converse}, we see that for a channel $\mathcal{N}^{A\to B}$ and for all protocols  $\mathcal{P} \in\mathcal{S}^{\rho^A} (\mathcal{N}, \eps)$, it holds that
\[
\begin{aligned}
\log \vert M \vert &\leq D_H^{\eps}\left(\mathcal{N}(\tau^{AB'})~||~\tau^{B'}\otimes \mathcal{N}(\rho^A)^B \right) \\ & {= I_H^\epsilon\left( B:B^\prime\right)_{\mathcal{N}(\tau^{AB^\prime})}}.
\end{aligned}
\]
\end{corollary}
\section{Proof of Proposition \ref{prop:IMstatechainrule}}\label{sec:ChainRules}
In this section, we present the proof of \cref{prop:IMstatechainrule}, which we restate below as a theorem:
\begin{theorem}
Given an  $\mathsf{IM}$-state $\rho^{A_fB_fC}$, it holds that
\[
I_H^{O({\eps})}(A_fB_f:C) ~\geq I_H^{O(\eps^{2})}(A_f: C) + I_H^{O(\eps^{2})}(B_f : A_fC) + O(\log \eps).
\]
\end{theorem}
\begin{proof}
First, define $\mathcal{N}^{AB\to C}$ to be the $\mathsf{IM}$-extended channel of the triple
\[
\left(\rho^{A_fB_fC}, \varphi_1^{A_fA}, \varphi_2^{B_fB}\right)
\]
where $\varphi_1^{A_fA}$ and $\varphi_2^{B_fB}$ are purifications of the states $\rho^{A_f}$ and $\rho^{B_f}$ respectively. The existence of this channel is guaranteed by \cref{lem:SmChannel}. Recall from \cref{sec:setup} that 
$
\mathcal{S}(\mathcal{N}, \delta)$
denotes the set of all those protocol $\mathcal{P}$ which make an (average) error at most $\delta$ while sending classical messages through the channel $\mathcal{N}^{AB\to C}$, with the help of shared entanglement. Note that now our channel takes as input the states in the bipartite system $AB$ and sends the output to the system $C$. Thus, the description of any protocol $\mathcal{P}$ for this channel will be given by the tuple:
\[
\left(M, \mathcal{N}, \mathcal{E}^{M_{AB}E_{AB}\to AB}, \mathcal{D}^{E_CC\to \widehat{M}_{AB}}, \ket{\varphi}^{E_{AB}E_C}\right)
\]
It is important to note that the above description does not treat the two systems $AB$ as belonging to two different senders. This allows us to bound the rate at which \emph{any} protocol can send classical messages through $\mathcal{N}$, and not just those protocols which treat $A$ and $B$ as belonging to two spatially separated senders. Also, recall that 
$
\mathcal{S}^{\varphi_1^{A}\otimes \varphi_2^{B}}(\mathcal{N},\delta)
$
denotes that subset of $\mathcal{S}(\mathcal{N},\delta)$ which contains the protocols $\mathcal{P}$ for which the state created by the encoder $\mathcal{E}^{M_{AB}E_{AB}\to AB}$ on the system $AB$ is $
\varphi_1^A\otimes \varphi_2^B $.
Then, by \cref{lem:converse} and \cref{corol:usableconverse}, we know that, for \emph{all} protocols in the set $\mathcal{S}^{\varphi_1^{A}\otimes \varphi_2^{B}}(\mathcal{N},\delta)$ it holds that
\[
\begin{aligned}
\log M \leq &~ D_H^{\delta}\left(\mathcal{N}^{AB\to C}\left(\varphi_1^{AA_f}\otimes\varphi_2^{BB_f}\right)~\Big|\Big|\right. \\ &~\left.\varphi_1^{A_f}\otimes \varphi_2^{B_f}\otimes \mathcal{N}\left(\varphi_1^{A}\otimes \varphi_2^{B}\right)\right) \\
= &~ D_H^{\delta}\left(\rho^{A_fB_fC}~||~\rho^{A_f}\otimes \rho^{B_f}\otimes \mathcal{N}\left(\varphi_1^{A}\otimes \varphi_2^{B}\right)\right) \\
= &~D_H^{\delta}\left(\rho^{A_fB_fC}~||~\rho^{A_fB_f}\otimes \rho^C\right) \\
= &~I_H^{\delta}(A_fB_f:C)
\end{aligned}
\]
We will now exhibit a protocol for sending classical information with entanglement assistance through the channel $\mathcal{N}$ which achieves the rate
\[
I_H^{\eps^{\prime}}(A_f: C) + I_H^{\eps^{\prime \prime}}(B_f : A_fC) + O(\log \eps)
\]
and also satisfies the invariant that the state the encoder of this protocol creates at the input to the channel, averaged over all messages, is
\[
\varphi_1^{A}\otimes \varphi_2^{B}
\]
This protocol makes an error $O(\sqrt{\eps})$. Thus, by setting 
\[
\delta\gets O(\sqrt{\eps})
\]
and noticing that this protocol belongs to the set
\[
\mathcal{S}^{\varphi_1^{A}\otimes \varphi_2^{B}}(\mathcal{N}, O(\sqrt{\eps})),
\]
we will conclude that,
\begin{align*}
I_H^{O(\sqrt{\eps})}(A_fB_f:C) &~\geq I_H^{O(\eps)}(A_f: C) + I_H^{O(\eps)}(B_f:A_fC) + O(\log \eps).
\end{align*}

Replacing $\epsilon \leftarrow\sqrt{\epsilon}$ gives the expression in the required form.
\\ \vspace{2mm} \\
\noindent \textbf{The Protocol} 

\vspace{1mm} 
To describe the protocol, it will be easier to consider two senders Alice and Bob who have access to the systems $A$ and $B$ of the channel $\mathcal{N}$ respectively. We also refer to the receiver as Charlie. Set
\[
\begin{aligned}
R_1 &\gets I_H^{\eps^{\prime}}(A_f:C) + O(\log \eps) \\
R_2 &\gets I_H^{\eps^{\prime\prime}}(B_f:A_fC)+O(\log \eps)
\end{aligned}
\]
Resources:
\begin{enumerate}
    \item Alice possesses a set of messages $[M]$ of size, where $M=2^{R_1}$.
  
    \item She also shares $2^{R_1}$ copies of the state $\ket{\varphi_1}^{AA_f}$ with Charlie:
    \[
    \ket{\varphi_1}^{E_{A_1}E_{C_1}}    \ket{\varphi_1}^{E_{A_2}E_{C_2}}\ldots     \ket{\varphi_1}^{E_{A_{2^{R_1}}}E_{C_{2^{R_1}}}}
    \]
    where 
    \[
    A\equiv E_{A_i}\textup{ and }A_f\equiv E_{C_i}
    \]
    for all $i$.
    \item Bob possesses a set of messages $[N]$ of size
    \[
    N\gets 2^{R_2}.
    \]
    \item He also shares $2^{R_2}$ copies of the state $\ket{\varphi_2}^{BB_f}$ with Charlie:
    \[
    \ket{\varphi_2}^{F_{B_1}F_{C_1}}    \ket{\varphi_2}^{F_{B_2}F_{C_2}}\ldots     \ket{\varphi_2}^{F_{B_{2^{R_2}}}F_{C_{2^{R_2}}}}
    \]
    where 
    \[
    B\equiv F_{B_j}\textup{ and }B_f\equiv F_{C_j}
    \]
    for all $j$.
\end{enumerate}
To describe the protocol in terms of the notation that we defined in \cref{EA_Classical_Setup}, consider the following assignments:\\
\[
\psi^{MM_{AB}}\gets \sum\limits_{\substack{i\in [M] \\ j\in [N]}}\frac{1}{MN} \ketbra{m,n}^{MN}\otimes \ketbra{m,n}^{M_AN_B}.
\]
and
\[
\ket{\varphi}^{E_{AB}E_C}\gets \left(\bigotimes\limits_{i\in [M]} \ket{\varphi_1}^{E_{A_i}E_{C_i}}\right)~\otimes~  \left(\bigotimes\limits_{j\in [N]} \ket{\varphi_2}^{F_{B_j}F_{C_j}}\right).
\]
\vspace{1mm}\\
One should also note that the encoder of the protocol $\mathcal{E}$ acts on the systems:
\[
M_AN_BE_{A_1}\ldots E_{A_{M}}F_{B_1}\ldots F_{B_{N}} \to A
\]
and the decoder $\mathcal{D}$ acts on
\[
CE_{C_1}\ldots E_{C_M}F_{C_1}\ldots F_{C_N}\to \widehat{M}\widehat{N} .
\]

We will now give a brief and informal overview of the design of the encoder $\mathcal{E}$ and the decoder $\mathcal{D}$. A detailed description along with the error analysis is provided in Appendix \ref{sec:multiparty} and \ref{sec:decodingProcedure}.
\begin{enumerate}
    \item To send the message $m\in [M]$, Alice inputs the contents of the register $E_{A_m}$ into the system $A$.
    \item To send the message $n\in [N]$, Bob inputs the contents of the register $F_{B_n}$ into the system $B$.
    \item To decode, Charlie first disregards the input from Bob as noise and decodes only for Alice.
    \item Having successfully decoded Alice's message, Charlie then uses this as side information to decode Bob's message at a higher rate.
\end{enumerate}
It is not hard to see that, for any the channel $\mathcal{N}$ from Alice to Charlie, while averaging over Bob's input can be considered to be:
\[
\mathcal{N}_0^{A\to C}\left(\cdot\right)\coloneqq\mathcal{N}^{AB\to C} \left(\cdot \otimes \varphi_2^{B}\right).
\]
With an analysis similar to $\mathcal{N}_0$, we can show that the rate of communication for $\mathcal{N}_1$ is
\begin{equation} \label{eq:R_1}
    I_H^{\eps^{\prime}}(A_f : C)+O(\log \eps).
\end{equation}
Since Anshu-Jain-Warsi protocol decodes correct $m$ with high probability, we can assume that Charlie knows $m$ while decoding Bob's message.
In other words, Charlie will decode Bob's message conditioned on Alice's message being $m$.
As usual, the actual state in the protocol may differ from the conditioned state, but the gentle measurement lemma guarantees that these states are not far, and the $\mathcal{L}_1$ distance between them can be consumed in the overall error of the protocol.
Since we assume that Alice's message was $m$, we define a new channel for analysis

\[
\mathcal{N}_1^{B\to C E_{C_m} }(\cdot )\coloneqq \mathcal{N}^{AB\to C}(\varphi_1^{AE_{C_m}}\otimes \cdot).
\]
Since the system $E_{C_{m}}\equiv A_f$ for all $m$, it holds that, conditioned on correct decoding for Alice, the rate at which Bob can communicate with Charlie is given by
\begin{equation} \label{eq:R_2}
    I_H^{\eps^{\prime\prime}}(B_f: CA_f)+O(\log \eps).
\end{equation}
Thus, the total rate of the protocol is given by adding expressions \eqref{eq:R_1} and \eqref{eq:R_2}, which is equal to $R_1+R_2$ by our choice.
This completes our proof sketch that there exists a strategy for achieving rates $R_1 + R_2$.
\end{proof}
\section{Chain Rules for General Quantum States}\label{sec:overviewfullthm}
In this section, we will introduce the ideas required to prove \cref{thm:mainthmintro}. We formally restate the theorem below:
\begin{theorem}\label{thm:formalmain}
Given $\eps>0$ and a tripartite state $\rho^{A_fB_fC}$, it holds that 
\[
I_H^{\eps}(A_fB_f:C) \geq I_H^{O(\eps^4)}(A_f:C)+I_H^{O(\eps^4)}(B_f:A_fC)-I_{\max}(A_f:B_f)-\log (1-O(\eps^{1/4}))-O(1).
\]
\end{theorem}
First, we will need the concept of \emph{quantum rejection sampling} as introduced in \cite{JRS'05}.

\subsection{Quantum Rejection Sampling}\label{sec:quantumrejectionsampling}
The rejection sampling problem can be framed as follows:
\begin{figure}[H]
\centering
\fbox{\parbox{\textwidth}{
\begin{enumerate}
    \item Consider two distributions $P_X$ and $Q_X$ over some alphabet $\mathcal{X}$, with the assumption that $\textup{supp}(P)\subseteq \textup{supp}(Q)$.
    \item Alice has access to iid samples from the distribution $Q_X$.
    \item The task is for Alice is to output a letter $X_{\textup{output}}$, using only the samples from $Q_X$ and her own private coins, such that $X_{\textup{output}}\sim P_X$.
    \item We also require that Alice uses as few iid samples from $Q_X$ as possible.
\end{enumerate}}}
\caption{Classical Rejection Sampling}
\end{figure}

 It can be shown that Alice can achieve the above task with $2^{D_{\max}(P_X~||~Q_X)}$ many samples on expectation. In this paper we will require the quantum analog of this problem, which can be stated as follows:

\begin{figure}[H]
    \centering
    \fbox{\parbox{\textwidth}{
\begin{enumerate}
       \item Consider two quantum states $\rho^A$ and $\sigma^A$ such that $\textup{supp}(\rho)\subseteq \textup{supp}(\sigma)$.
    \item Alice is provided multiple independent copies of the state $\sigma^A$ along with ancilla registers as workspace.
    \item The task is for her to produce the state $\rho^A$, while using as few copies of $\sigma^A$ as possible.
\end{enumerate}}}
\caption{Quantum Rejection Sampling}
\end{figure}

It can be shown \cite{JRS'05} that the above task can be achieved with $2^{D_{\max}(\rho^A~||~\sigma^A)}$ many copies of the state $\sigma^A$, on expectation. 
\begin{remark}
In fact, if Alice can tolerate some error in the state that she outputs, in the sense that she creates a state $\rho^{\prime A}\overset{\eps}{\approx}\rho^A$, then the task can be achieved with $\frac{1}{\eps}\cdot 2^{D_{\max}^{\eps}(\rho^A~||~\sigma^A)}$ copies of $\sigma^A$ on expectation. However, due to the nature of our protocol, we will require the exact version of this protocol, which requires more copies of $\rho^A$ to work.
\end{remark}
The way the protocol works is as follows:
\begin{enumerate}
    \item Alice possesses multiple iid copies of $\sigma^A$.
    \item By definition, it holds that
    \[
    \rho^A \leq 2^{D_{\max}^{\eps}(\rho~||~\sigma)}\sigma^A
    \]
    which implies that there exists a quantum state $\tau^A$ such that
    \[
    \sigma^A = \frac{1}{2^{D_{\max}}} \rho^A +\left(1-\frac{1}{2^{D_{\max}}}\right) \tau^A
    \]
    where in the above we used $D_{\max}$ as a shorthand for $D_{\max}(\rho^A~||~\sigma^A)$.
    \item Alice uses two registers $R$ and $Q$ to produce a certain purification of $\sigma^A$. Here, $Q$ will be a single qubit register:
    \[
    \ket{\sigma}^{ARQ}\coloneqq \sqrt{\frac{1}{2^{D_{\max}}}}\ket{\rho}^{AR}\ket{0}^Q+\sqrt{\left(1-\frac{1}{2^{D_{\max}}}\right) }\ket{\tau}^{AR}\ket{1}^Q
    \]
    where $\ket{\rho}^{AR}$ and $\ket{\tau}^{AR}$ are purifications of $\rho^A$ and $\tau^A$.
    \item Alice performs this purification for a large number of copies of $\sigma^A$.
    \item Alice then measures the $Q$ register in the computational basis. She gets $0$ with probability $1/2^{D_{\max}}$.
    \item Discarding the system $R$ completes the protocol.
\end{enumerate}

It is easy to see that the protocol detailed above requires $2^{D_{\max}}$ many copies of $\sigma^A$ to succeed on expectation. We will require this idea in the following sections.

\subsection{The Channel $\mathcal{N}^{AB\to C}$ for General States}

The Quantum Rejection Sampling protocol gives us a hint as to how we might go about defining a channel $\mathcal{N}^{AB\to C}$ and some state $\ket{\phi}^{A_fB_fAB}$ such that 
\[
\mathcal{N}^{AB\to C} \left(\phi^{A_fB_fAB} \right) = \rho^{A_fB_fC}
\]
for some fixed $\rho^{A_fB_fC}$. The idea is as follows:
\begin{enumerate}
    \item Consider the marginals $\rho^{A_f}$ and $\rho^{B_f}$ of $\rho^{A_fB_fC}$ and their purifications $\ket{\varphi_1}^{AA_f}$ and $\ket{\varphi_2}^{BB_f}$ as before.
    \item Let Alice have access to the $A$ and $B$ systems of multiple iid copies of $\ket{\varphi_1}^{AA_f}\ket{\varphi_2}^{BB_f}$.
    \item Recall that by definition,
    \[
    \rho^{A_fB_f}\leq 2^{I_{\max}(A_f:B_f)}\rho^{A_f}\otimes \rho^{B_f}
    \]
    where 
    \[
    I_{\max}(A_f:B_f)= D_{\max}(\rho^{A_fB_f}~||~\rho^{A_f}\otimes \rho^{B_f}).
    \]
    We will use the shorthand \imax~ to refer to $I_{\max}(A_f:B_f)$ from here onward.
    \item Consider then, the purification $\ket{\varphi}^{AA_fBB_fQ}$ of $\rho^{A_f}\otimes \rho^{B_f}$:
    \[
    \ket{\varphi}^{AA_fBB_fQ}\coloneqq \sqrt{\frac{1}{2^{\imax}}}\ket{\phi}^{AA_fBB_f}\ket{0}^Q+\sqrt{\left(1-\frac{1}{2^{\imax}}\right)}\ket{\tau}^{AA_fBB_f}\ket{1}^Q
    \]
    where $\ket{\phi}^{AA_fBB_f}$ is some purification of $\rho^{A_fB_f}$.
    \item Since Alice possesses the $A$ and $B$ systems of the state $\ket{\varphi_1}^{AA_f}\ket{\varphi_2}^{BB_f}$, she can use the Uhlmann isometry $W^{AB \to ABQ}$ to create the state $\ket{\varphi}^{AA_fBB_fQ}$. She does for many copies of the states that she possesses.
    \item Now Alice measures the $Q$ register for each copy of the state $\ket{\varphi}^{AA_fBB_fQ}$. On expectation, she receives a $0$ outcome after $2^{\imax}$ many measurements.
    \item Now, as before, consider an arbitrary purification $\ket{\rho}^{A_fB_fCE}$ of $\rho^{A_fB_fC}$. Then, there exists an Uhlmann isometry $V^{AB\to CE}$ such that
    \[
    V^{AB\to CE}\ket{\phi}^{AA_fBB_f}= \ket{\rho}^{A_fB_fCE}
    \]
\item Composing the trace out operation on the system $E$ with $V$ gives us the channel $\mathcal{N}^{AB\to C}$.
\end{enumerate}

The above discussion implies the following lemma:
\begin{lemma}\label{lem:channelgeneral}
Given any quantum state $\rho^{A_fB_fC}$, and an arbitrary purification $\ket{\phi}^{A_fB_fAB}$ of the state $\rho^{A_fB_f}$, there exists a channel $\mathcal{N}^{AB\to C}$ such that
\[
\I^{A_fB_f}\otimes \mathcal{N}^{AB\to C}\left(\phi^{A_fB_fAB}\right)=\rho^{A_fB_fC}.
\]
\end{lemma}

\subsection{Towards a Proof of Theorem \ref{thm:formalmain}}

In this section, we informally describe our strategy to prove Theorem \ref{thm:formalmain}. First, we fix a purification $\ket{\phi}^{AA_fBB_f}$ of $\rho^{A_fB_f}$ and consider the channel $\mathcal{N}^{AB\to C}$ given by \cref{lem:channelgeneral}. We construct an entanglement-assisted communication protocol for this channel using the shared entangled states $\ket{\varphi_1}^{AA_f}$ and $\ket{\varphi_2}^{BB_f}$, which are purifications of $\rho^{A_f}$ and $\rho^{B_f}$ respectively. To do this we make use of two things:
\begin{enumerate}
    \item The Quantum Rejection Sampling Algorithm.
    \item A completely dephasing channel $\mathcal{P}^{X_A\to X_B}$ from Alice to Bob, as an extra resource, which can send $I_{\max}(A_f:B_f)+\log \frac{1}{\delta}$ many bits noiselessly. Here $X_A$ and $X_B$ denote Alice and Bob's classical registers respectively, and both are of size $\frac{1}{\delta}2^{\imax}$.
\end{enumerate}
The rationale behind integrating rejection sampling is simple. 
Instead of using $[M]$ copies $\ket{\varphi_1}$, we share $C_0 M$ copies where $C_0$ is suitably chosen.
The whole protocol is now viewed as having $M$ blocks of size $C_0$.
The Rejection sampling then takes us to a candidate index (say $b^*$) with certain desired properties (with high probability).
The number of such coordinates having index $b^*$ is thus $M$ (one in each block).
Restricted over these coordinates, the protocol now has a behavior similar to that of $\mathsf{IM-state}$ protocol.
A similar modification is done for the second step of successive cancellation as well, where $C_0 N$ copies of $\ket{\varphi_2}$ are shared.
The detailed protocol is as follows:
\begin{enumerate}
\captionof{table}{Protocol Modified quantum assisted classical communication (with blocks)}\label{protocol:general}
    \item We arrange Alice's set of messages as $[M]\times [N]$, where
    \[
    \begin{aligned}
    \log M &= I_H^{\eps}(A_f : C) + O(\log \eps) \\
    \log N &= I_H^{\eps}(B_f :A_f  C) + O(\log \eps)    .
    \end{aligned}
    \]
    \item Alice shares $M\times \frac{1}{\delta}\cdot 2^{\imax}$ many copies of the state $\ket{\varphi_1}^{A_fA}$ with Bob. She divides these into $M$ blocks, each of size $\frac{2^{\imax}}{\delta}$.
     Blocks are indexed by $m \in [M]$ and the elements inside a block are further indexed by $b \in \mathcal{B} \coloneqq \left[\frac{2^{\imax}}{\delta}\right]$. 
     The corresponding states
are therefore represented as    \[
    \bigotimes \limits_{i\in [M]}\bigotimes \limits_{b\in \mathcal{B}} \ket{\varphi_1}^{A_{f_{b,i}}A_{b,i}}
    \]
    \item Similarly Alice shares $N\times \frac{1}{\delta}\cdot  2^{\imax}$ many copies of the state $\ket{\varphi_2}^{BB_f}$ with Bob.
    She divides these into $N$ blocks of size $\frac{1}{\delta}\cdot 2^{\imax}$ as well.
    The blocks are indexed analogously.
    The states  then are
    \[
    \bigotimes \limits_{j\in [N]}\bigotimes \limits_{b\in \mathcal{B}} \ket{\varphi_2}^{B_{f_{b,j}}B_{b,j}}
    \]
    \item To send a message $(m,n)$  Alice picks the $m$-th block of $\ket{\varphi_1}$'s 
    \[
    \bigotimes \limits_{b\in \mathcal{B}} \ket{\varphi_1}^{A_{f_{b,m}}A_{b,m}}
    \]
    and the $n$-th block of $\ket{\varphi_2}$'s
    \[
    \bigotimes \limits_{b\in \mathcal{B}} \ket{\varphi_2}^{B_{f_{b,n}}B_{b,n}}.
    \]
    \item For each $b\in \mathcal{B}$, Alice applies the isometry $W^{AB\to ABQ}$ such that
    \[
    W\ket{\varphi_1}^{A_{f_{b,m}}A_{b,m}}\ket{\varphi_2}^{B_{f_{b,n}}B_{b,n}} = \ket{\varphi}^{A_{b,m}A_{f_{b,m}}B_{b,n}B_{f_{b,n}}Q_b}_b
    \]
    where $\ket{\varphi}^{A_{b,m}A_{f_{b,m}}B_{b,n}B_{f_{b,n}}Q_b}_b $ indicates the $b$-th copy of $\ket{\varphi}^{AA_fBB_fQ}$ .
    \item Alice then measures the registers $Q_1Q_2\ldots Q_{\frac{1}{\delta}2^{\imax}}$ in a \emph{random order}, and stops the first time the measurement succeeds. By \cref{claim:highprobability} she gets at least one $0$ outcome with probability at least $1-e^{-1/\delta}$. If none of the measurements succeed, Alice aborts.
    \item Suppose the index on which the measurement succeeded is $b^*$. Then, by \cref{claim:distribution}, the distribution of $b^*$ is uniform. Alice then sends the index $b^*$ through the noiseless completely dephasing channel to Bob. \label{point:distribution}
    \item Alice also puts the contents of the system $A_{b^*,m}B_{b^*,n}$ into the system $AB$, i.e., the systems which are input to the channel. She initialises the $A_{b^*,m}B_{b^*,n}$ registers with some junk state.
    \item Bob can then simply pick out the $b^*$-th element in every message block and repeat the successive cancellation decoding procedure as given in \cref{sec:multiparty} (see \cref{claim:generaldecoding} for details). To be precise, Bob performs his measurements on the collective states
    \[
    \begin{aligned}
    &\bigotimes\limits_{\substack{i\in [M], j\in [N]\\ i\neq m, j\neq n} }\varphi_1^{A_{f_{b^*,i}}}\otimes \varphi_2^{B_{f_{b^*,j}}}\bigotimes \mathcal{N}\left(\phi^{ABA_{f_{b^*,m}}B_{f_{b^*,n}}}_{b^*}\right) \\
    =~& \bigotimes\limits_{\substack{i\in [M], j\in [N]\\ i\neq m, j\neq n} }\rho^{A_{f_{b^*,i}}}\otimes \rho^{B_{f_{b^*,j}}}\bigotimes \rho^{CA_{f_{b^*,m}}B_{f_{b^*,n}}}
    \end{aligned}
    \]
\end{enumerate}
where 
\[
\rho^{CA_{f_{b^*,m}}B_{f_{b^*,n}}}\equiv \rho^{A_fB_fC}.
\]

The above protocol gives us an achievable strategy to send $I_{H}^{\eps}(A_f:C)+I_H^{\eps}(B_f:A_fC)$ many bits via the channel $\mathcal{N}^{AB\to C}\otimes \mathcal{P}^{X_A\to X_B}$, while making an overall decoding error of at most $28\sqrt{\eps}$ (see \cref{claim:generaldecoding} for details). To complete the proof we must find a suitable upper bound as given by \cref{corol:usableconverse}. Before using \cref{corol:usableconverse} however, we should point out some subtle issues:
\begin{enumerate}
    \item The proof of \cref{corol:usableconverse} does not consider encoders that can abort the protocol. However, since the quantum rejection sampling procedure can fail with some non-zero probability, we must ensure that \cref{corol:usableconverse} can be suitably adapted to this case. In \cref{claim:extendedreduction} we extend the proof of \cref{corol:usableconverse} to the case when the encoder can toss its own private coins and may abort the protocol with some probability. We show that the results of \cref{corol:usableconverse} essentially remain unchanged even in this case. \label{point:converse}
    \item Recall that \cref{corol:usableconverse} provides an upper bound on the number of bits any protocol can send through a channel as a function of the state that the encoder of the protocol creates on the input register of the channel, averaged over all messages. For the protocol that we presented in this section, we must find this averaged input state on the system $ABX_A$. By the arguments in \cref{point:converse} above, we are only interested in the case when Alice does not abort. Conditioned on Alice not aborting, the state created on the input system $AB$ of the channel is $\rho^{AB}$. Note that this state is independent of the index $b^*$ on which the measurement succeeded. Also note that, by \cref{point:distribution}, the distribution on the system $X_A$, which is input to the completely dephasing channel $\mathcal{P}^{X_A\to X_B}$ is uniform over the size of $X_A$. Therefore, the state, averaged over all other systems, on the input registers $ABX_A$ of the channel $\mathcal{N}\otimes \mathcal{P}$ is 
    \[
    \phi^{AB}\otimes \frac{\I^{X_A}}{\frac{1}{\delta}2^{\imax}}.
    \]
    where $\phi^{AB}$ is the marginal of the state $\ket{\phi}^{A_fB_fAB}$.\label{point:input}
 \end{enumerate}
For ease of notation, let us refer to Protocol \ref{protocol:general} as $\mathcal{P}_{\textsc{achievable}}$. Let us denote by $E$ the event that the quantum rejection sampling phase succeeds, and define 
\[
\mathcal{P}_{\textsc{achievable}}\vert_{E}
\]
be the execution of Protocol \ref{protocol:general} conditioned on the event that the quantum rejection sampling phase succeeded. Then, by the discussion above, one can see that
\[
\mathcal{P}_{\textsc{achievable}}\vert_{E} \in \mathcal{S}^{\rho^{AB}\otimes \frac{\I^{X_A}}{\frac{1}{\delta}2^{\imax}}}\left(\mathcal{N}^{AB\to C}\otimes \mathcal{P}^{X_A\to X_B}, 28\sqrt{\eps}\right).
\]
Then, by the arguments of \cref{claim:extendedreduction}, we see that the maximum number of bits that can be transmitted by Protocol \ref{protocol:general} is 
\[
D_{\textsc{final}}\coloneqq D_H^{28\sqrt{\eps}}\left(\mathcal{N}^{AB\to C}\otimes \mathcal{P}^{X_A\to X_B}\left(\phi^{A_fB_fAB}\otimes \Phi^{RX_A}\right)~||~\mathcal{N}(\phi^{AB})\otimes \mathcal{P}\left(\pi^{X_A}\right)\otimes \phi^{A_fB_f} \otimes \pi^{R}\right)
\]
where $\Phi^{RX_A}$ is a maximally entangled state on the system $X_A$ and $R\cong X_A$ is the system purifying the maximally mixed state on the system $X_A$. We use the notation $\pi^{X_A}$ to denote the maximally mixed state on $X_A$. Recall that we can make the above statement by \cref{claim:extendedreduction} and \cref{corol:usableconverse}, and the fact that the converse given by those results is true for any arbitrary purification of $\phi^{AB}\otimes \pi^{X_A}$. Note that
\[
\mathcal{P}^{X_A\to X_B}\left(\Phi^{RX_A}\right)=\frac{1}{\abs{X_A}}\sum\limits_{x}\ketbra{x}^{R}\otimes \ketbra{x}^{X_B}.
\]
and
\[
\mathcal{P}^{X_A\to X_B}\left(\pi^{X_A}\right)= \pi^{X_B}.
\]
Then, by \cref{claim:IHgeneral}, we can see that
\begin{align*}
    D_{\textsc{final}}\leq D_H^{\sqrt{28}\eps^{1/4}}\left(\mathcal{N}^{AB\to C}\left(\phi^{ABA_fB_f}\right)~||~\mathcal{N}^{AB\to C}\left(\phi^{AB}\right)\otimes \phi^{A_fB_f}\right)+\log \abs{X_A}-\log (1-O(\eps^{1/4})).
\end{align*}
Now recall that
\begin{align*}
    &\mathcal{N}^{AB\to C}\left(\phi^{ABA_fB_f}\right)= \rho^{A_fB_fC} \\
    &\mathcal{N}^{AB\to C}\left(\phi^{AB}\right)= \rho^{C} \\
    &\phi^{A_fB_f}=\rho^{A_fB_f} \\
    &\log \abs{X_A}= I_{\max}(A_f:B_f)_{\rho^{A_fB_f}}+\log\frac{1}{\delta}.
\end{align*}
Collating all these arguments together, we see that
\begin{align*}
    I_H^{\eps}(A_F:C)+I_H^{\eps}(B_F:A_FC) &\leq D_{\textsc{final}} \\
    &\leq D_H^{\sqrt{28}\eps^{1/4}}(\rho^{A_fB_fC}~||~\rho^{A_fB_f}\otimes \rho^C)+I_{\max}(A_f:B_f)+\log (1-O(\eps^{1/4}))+\log \frac{1}{\delta} \\
    &= I_{H}^{\sqrt{28}\eps^{1/4}}(A_fB_f:C)+I_{\max}(A_f:B_f)+\log (1-O(\eps^{1/4}))+\log \frac{1}{\delta}.
\end{align*}
Finally, we rearrange terms in the above inequality, while setting $\eps\gets \eps^4$ and using the fact that $\log\frac{1}{\delta}=O(1)$. This concludes the proof of \cref{thm:formalmain}.

\section*{Acknowledgement}
We would like to thank Pranab Sen for several helpful discussions and suggestions. SC would like to acknowledge support from the National Research Foundation, including under NRF RF Award No. NRF-NRFF2013-13 and NRF2021-QEP2-02-P05 and the Prime Minister’s Office, Singapore and the Ministry of Education, Singapore, under the Research Centres of Excellence program.
U.~K. is supported by the Singapore Ministry of Education and the National Research Foundation through the core grants of the Centre for Quantum Technologies

\newpage
    \bibliography{chain_rules_ref}
    \bibliographystyle{alpha}
    
 \newpage
    
    \appendix
\section{Proof of Lemma \ref{lem:converse}}\label{appendix:proof of lem:converse}
\begin{proof}
Recall that the state, after Alice encodes, is given by,
\[ \rho^{ME_BA} = \mathcal{E}~\left(\psi^{MM_A}\otimes \varphi^{E_AE_B}\right). \]
where $\psi$ and $\varphi$ are as per protocol~\ref{EA_Classical_Setup}.
\noindent Consider an arbitrary purification $\ket{\hat{\tau}}^{ME_BAF}$ of $\rho^{ME_BA}$ and define
\begin{align}
    &\rho^{ME_BBF}\coloneqq  \mathcal{N}^{A\to B}\left(\ketbra{\hat{\tau}}^{ME_BAF}\right) \\
    &\phi^{MM'}\coloneqq \mathbb{I}_M \otimes \mathcal{D}\left(\rho^{MBE_B}\right). \label{eq:decoding_expression}
\end{align}
Since $\mathcal{P} \in \mathcal{S}^{\rho^A}(\mathcal{N},\eps)$, it holds that
\[
\begin{aligned}
&\Tr_{ME_B}\left( \rho^{ME_BA}\right)= \rho^A\\
&\phi^{M}=\rho^M=\psi^M=\frac{\I}{M} \\
&\begin{aligned}\er(\mathcal{P}) &=\norm{\phi^{MM'} - \frac{1}{M}\sum_m\ketbra{m}^M\otimes \ketbra{m}^{M'}}_1 \\ &\leq \eps.\end{aligned}
\end{aligned}
\]
Consider the projector
\[
\Pi^{MM'}\coloneqq \sum_m\ketbra{m}^M\otimes \ketbra{m}^{M'}.
\]
It is then easy to see that
\[
\Tr\left[\Pi^{MM'}\phi^{MM'}\right]\geq 1-\eps.
\]
Thus, since $\phi^{MM'}$ satisfies all the conditions of \cref{fact:closeness}, we see that
\[
\log \vert M \vert \leq D_H^{\eps}(\phi^{MM'}~||~\phi^M\otimes \mathcal{D}\left(\sigma^B\otimes \rho^{E_B})\right).
\]
Each of the following inequality is a straightforward application of the data processing inequality.

\begin{align*}
\log \vert M \vert &\leq D_H^{\eps}\left(\phi^{MM'}~||~\phi^M\otimes \mathcal{D}\left(\sigma^B\otimes \rho^{E_B}\right)\right) \\
&\leq D_H^{\eps}\left(\rho^{MBE_B}~||~\rho^M\otimes \sigma^B\otimes \rho^{E_B}\right) \\ \intertext{$\ldots$ From Eq.~\eqref{eq:decoding_expression} and $\rho^M= \phi^M$} 
&= D_H^{\eps}\left(\rho^{MBE_B}~||~\rho^{ME_B}\otimes \sigma^B\right) \\
& \leq D_H^{\eps}\left(\rho^{MBE_BF}~||~\rho^{ME_BF}\otimes \sigma^B\right) \\
&= D_H^{\eps}\left(\mathcal{N}(\hat{\tau}^{MAE_BF})~||~\hat{\tau}^{ME_BF}\otimes \sigma^B\right).
\end{align*}
The first equality follows from $\rho^{ME_B}= \rho^M \otimes \rho^B$; whereas the last equality follows from $\hat{\tau}= \mathcal{N}(\rho)$ and noting that $\mathcal{N}$ does not act on any of the registers $ME_BF$.
Thus, $\rho^{ME_BF}= \hat{\tau}^{ME_BF}$.

\noindent Now, one can consider any purification $\tau^{AB^\prime}$ of $\rho^A$.
By Uhlamann's theorem, there exists an isometry $V^{ME_BF \rightarrow B^\prime}$ such that $V^{ME_BF \rightarrow B^\prime} (\hat{\tau}) = \tau$.   
It then follows that
 \[
 \begin{aligned}
 &D_H^{\eps}\left(\mathcal{N}({\tau}^{AB^\prime})~||~{\tau}^{B^\prime}\otimes \sigma^B\right) = ~& D_H^{\eps}\left(\mathcal{N}(\hat{\tau}^{MAE_BF})~||~\hat{\tau}^{ME_BF}\otimes \sigma^B\right) \geq \log \vert M \vert.
 \end{aligned}
 \]
This concludes the proof.
\end{proof}

\section{Achievable Strategies}\label{sec:AJW}
\subsection{The Anshu-Jain-Warsi Protocol}

In this section, we will recall the one-shot entanglement-assisted classical message transmission protocol due to Anshu, Jain and Warsi~\cite{quantum_assistated_classical_anurag}, which we abbreviate as the AJW protocol. The protocol proceeds as follows:

\begin{enumerate}
    \item We are given a point to point channel $\mathcal{N}^{A\to B}$ and a starting state $\psi^{MM_A}$ held by the sender Alice;
    \[ \psi^{MM_A}= \dfrac{1}{2^R}\sum_{m\in [2^R]} \ketbra{m}^M\otimes \ketbra{m}^{M_A}. \]
    \item Sender Alice and receiver Bob share $2^R$ copies of some pure state $\ket{\varphi}^{E_AE_B}$ as follows,
    \[
    \ket{\varphi}^{E_{A_1}E_{B_1}}\ket{\varphi}^{E_{A_2}E_{B_2}}\ldots\ket{\varphi}^{E_{A_{2^R}}E_{B_{2^R}}}
    \]
    where the systems  $E_{A_i}$ belong to Alice and $E_{B_i}$ belong to Bob.
    \item Let Alice prepare some junk state $\ket{\textsc{junk}}^{A}$, where the system $A$ is isomorphic to $E_A$.
    \item The classical message that Alice wants to send is stored in the register $M_A$. 
    Suppose Alice wants to send a message $m$.
    Then Alice swaps systems $E_{A_m}$ and $A$, followed by the action of $\mathcal{N}^{A \rightarrow B}$.
    To be precise, Alice acts the controlled unitary
    \[
    \sum\limits_{m\in [2^R]} \ketbra{m}^{M_A}\otimes \textsc{SWAP}^{E_{A_m}A}
    \]
    on the systems $M_AE_{A_1}E_{A_2}\ldots E_{A_{2^R}}A$.
    And then applies the channel $\mathcal{N}^{A \rightarrow B}$.
    \item Let $\Pi$ be an optimal tester for $\left({I, \epsilon,  \varphi^{AE_B}, \mathcal{N}}\right)$ (with $\varphi^{AE_B} \equiv \varphi^{E_A E_B}$).
    Consider a set (indexed by $m$) of projectors acting jointly on the registers $B E_{B_1} E_{B_2} \ldots E_{B_{2^R}}$;
    \[
    \begin{aligned}
    \Lambda_m = &\I^{E_{B_1}} \otimes \cdots \otimes \I^{E_{B_{m-1}}} \otimes \Pi^{BE_{B_m}} \otimes \I^{E_{B_{m+1}}} \otimes \cdots \otimes \I^{E_{B_{2^R}}}.
    \end{aligned}\]
    Furthermore, using $\lbrace \Lambda_m\rbrace_{m \in [2^R]}$ define a POVM as follows:
    \[\Omega_m= \left(\sum_{i} \Lambda_i \right)^{-\frac{1}{2}} \Lambda_m \left(\sum_{i} \Lambda_i \right)^{-\frac{1}{2}}.\]
    This is the standard \emph{PGM} construction of $\Omega_m$ out of $\Lambda_m$.
    To decode, Bob simply measures with the POVM $\lbrace\Omega_m\rbrace_m$.
    The output of the POVM is represented by $\widehat{M}$ and the state at the end of the protocol is denoted by $\Theta_{\mathsf{END}}$.
\end{enumerate}
The above protocol has an error at most $\epsilon$~\cite{quantum_assistated_classical_anurag}, stated by the fact below.
\begin{fact}~\cite{quantum_assistated_classical_anurag} For any 
\[R \leq  I_{H}^\epsilon (E_B: B)_{\mathcal{N}\left(\ketbra{\varphi}^{AE_B}\right)} - 2 \log\left( \frac{1}{\epsilon}\right),\]
 where  \[\ket{\varphi}^{AE_B} \equiv \ket{\varphi}^{E_A E_B},\]
we have, \[\mathsf{Pr}\left( \widehat{M} \neq m \vert M=m \right)_{\Theta_{\mathsf{END}}} \leq 16 \epsilon.\] 
More formally, for all $i \in [m]$,
let \[\Theta_m = \bigotimes\limits_{i \neq m}\varphi^{E_{B_i}} \otimes \mathcal{N}^{A \rightarrow B} \left( \varphi^{A E_{B_m}}\right).\] 
Then,
$\Tr \left( \Omega_m \Theta_m \right) \geq 1 -16 \epsilon$. \label{fact:Anshu_achievability}
\end{fact}
Before going forward, we will need one additional observation about this protocol, which we state in the claim below.
The proof of the claim is fairly straightforward and follows from construction.
We include it for the sake of completeness. 
\begin{claim}
The state produced by the protocol above on register $A$ (the input register for the channel), averaged over all messages, is $\varphi^{A}$.
\end{claim}
\begin{proof}
Firstly, recall that Alice will act her encoder on the $M_A$ register of the state
\[
\psi^{MM_A}= \dfrac{1}{2^R} \sum_{m\in [2^R]} \ketbra{m}^M\otimes \ketbra{m}^{M_A}
\]
and the systems $E_{A_1}\ldots E_{A_{2^R}}$ of the shared states
\[
 \ket{\varphi}^{E_{A_1}E_{B_1}}\ket{\varphi}^{E_{A_2}E_{B_2}}\ldots\ket{\varphi}^{E_{A_{2^R}}E_{B_{2^R}}}.
\]
It is easy to see that after the encoding, the global state on all systems is as follows:
\[
\begin{aligned}
\frac{1}{2^R}\sum\limits_{m\in [2^R]}&\ketbra{m}^{M}\otimes \ketbra{m}^{M_A}\otimes \left(\ketbra{\varphi}^{AE_{B_m}}\right)\otimes  &\bigotimes\limits_{i\neq m}\ketbra{\varphi}^{E_{A_i}E_{B_i}}\otimes \ketbra{\textsc{junk}}^{E_{A_m}}.
\end{aligned}
\]
Tracing out all the registers except $A$, we see that the marginal on register $A$ is
$
\varphi^A$.
This proves the claim.
\end{proof}

\subsection{A Multi-Party Generalisation}\label{sec:multiparty}
Consider the following scenario:
Given a channel $\mathcal{N}^{AB\to C}$, and pure states $\ket{\varphi_1}^{E_AE_C}$ and $\ket{\varphi_2}^{F_B F_C}$.
Let the sender Alice and receiver Charlie share $2^{R_1}$ copies of $\ket{\varphi_1}$ as
\[
\ket{\varphi_1}^{E_{A_i}E_{C_i}}, \text{ where } i \in \left[ 2^{R_1}\right].
\]
where Alice possesses the registers $E_{A_i}$ and Charlie possesses the systems $E_{C_i}$. Similarly, the second sender Bob and receiver Charlie share $2^{R_2}$ copies of the state $\ket{\varphi_2}^{F_B F_C}$ as
\[
\ket{\varphi_2}^{F_{B_i}F_{C_i}} \text{ where } i \in \left[ 2^{R_2}\right].
\]
where Bob possesses the systems $F_{B_i}$ and Charlie the systems $F_{C_i}$.
To send the message pair $(m,n)\in [2^{R_1}]\times [2^{R_2}]$, Alice and Bob do the following protocol:

\begin{enumerate}
    \item Alice prepares a junk state in the system $A$, as $\ket{\textsc{junk}}^{A}$ and similarly Bob prepares $\ket{\textsc{junk}}^{B}$.
    \item Alice swaps the contents of $A$ with $E_{A_m}$ and Bob swaps the contents of $B$ with $F_{B_n}$.
    These operations can be expressed by the following controlled unitary maps:
\begin{align*}
   \mathsf{Enc}_{A} =   \sum\limits_{m\in [2^{R_1}]} \ketbra{m}^{M_A}\otimes \textsc{SWAP}^{E_{A_m}A}
 \\
    \mathsf{Enc}_B = \sum\limits_{n\in [2^{R_2}]} \ketbra{n}^{M_B}\otimes \textsc{SWAP}^{F_{B_n}B}
\end{align*}
   
    \item The senders then send the systems $AB$ through the channel.
\end{enumerate}

\subsection{The Decoding Procedure} \label{sec:decodingProcedure}
\begin{figure}[hbtp]
    \centering
    \fbox{\parbox{\textwidth}{
    Charlie performs the decoding in two phases:
\begin{enumerate}
    \item \label{item:alice_decoding} First where Charlie decodes Alice's message assuming \emph{nothing} about Bob's message.
In this step, Charlie outputs a candidate $\hat{m}$, for Alice's message. To do this, he uses a POVM $\brak{\Omega_{1,m}}_m$.
\item In the second step, Charlie outputs a candidate message $\hat{n}$ for Bob's message, assuming that Alice had sent $\hat{m}$. 
For the decoder, we need to define two POVMs, one each for outputting $\hat{m}$ and $\hat{n}$. He does this using a POVM $\brak{\Omega_{2,n}}_n$.
\item The POVMs $\brak{\Omega_{1,m}}$ and $\brak{\Omega_{2,n}}$ are defined explicitly later.
\end{enumerate}}}    \caption{Multiparty Decoding}    \label{fig:multiparty_decoding}
\end{figure}

\begin{claim} \label{claim:final_decoding}
    For any \[
    \begin{aligned}
    R_1  &\leq  I_{H}^\epsilon (E_C: C)_{\mathcal{N}\left(\varphi_1^{AE_C}\otimes \varphi_2^{BF_C}\right)} - 2 \log\left( \frac{1}{\epsilon}\right), \\
  R_2  &\leq  I_{H}^\epsilon (F_C: C E_C)_{\mathcal{N}\left(\varphi_1^{AE_C}\otimes \varphi_2^{BF_C}\right)} - 2 \log\left( \frac{1}{\epsilon}\right)
    \end{aligned}
    \]    
    \hspace{2cm}\text{ where } \[ \begin{aligned}
    &\ket{\varphi_1}^{AE_C} \equiv \ket{\varphi}^{E_A E_C}, \\
    &\ket{\varphi_2}^{BE_C} \equiv \ket{\varphi}^{F_B F_C}
    \end{aligned}
    \]
we have, \[\mathsf{Pr}\left[\left(\widehat{M},\widehat{N}\right) \neq \left(m,n\right) \vert \left(M, N\right)=\left(m, n\right) \right] \leq 28  \sqrt{\epsilon}.\] 
\end{claim}
We defer the proof of this claim to a later point.
The proof directly follows from Claim~\ref{claim:Bob_decoding_correctly} which itself uses Lemma~\ref{claim:Alice_decoding_Correctly} as an intermediate step.
Throughout the analysis, we will now assume that $R_1$ and $R_2$ satisfy the conditions stated by the hypothesis of Claim~\ref{claim:final_decoding}.

The next section is devoted to proving the above claim. We first focus on Charlie's decoding strategy for Alice and then on his decoding strategy for Bob.

\subsubsection{Decoding Alice}

\subsubsection*{Defining the POVM $\brak{\Omega_{1,m}}$:}

Consider \[
    I_H^{\eps}(E_C:C)_{\I^{E_C F_C}\otimes\mathcal{N^{AB \rightarrow C}}\left(\varphi_1^{AE_C}\otimes \varphi_2^{BF_C}\right)} \hspace{2cm}\]

    Let $\Pi_1^{E_CC}$ denote an optimal measurement for the above quantity.
    That is,
 
\begin{align} \label{eq:pi_1_ecc_a}
    2^{- I^{\epsilon}_{H}(E_C:C)_{\mathcal{N}\left(\varphi_1^{AE_C}\otimes \varphi_2^{BF_C}\right)}}     &=   \Tr \left[\Pi_1^{E_CC }~ \left(\mathcal{N}\left( \varphi_1^A \otimes \varphi_2^B \right) \otimes \varphi_1^{E_C} \right)\right] \\    
  \Tr\left[\Pi_1^{EcC}\left(\mathcal{N}(\varphi_1^{AE_C} \otimes \varphi_2^{B})\right)\right] & \geq 1-\eps. \label{eq:pi_1_ecc_b} \end{align}
Let \[
    \begin{aligned}
        \Lambda_{1,m}&= \I^{E_{C_1}} 
        \otimes \cdots \otimes \I^{E_{C_{m-1}}} \otimes  \cdots \otimes\Pi^{CE_{C_m}} \otimes \I^{E_{C_{m+1}}} \otimes  \cdots \otimes \I^{E_{C_{2^{R_1}}}} ,  \\
    \Omega_{1,m} &= \left(\sum_{i} \Lambda_{1,i} \right)^{-\frac{1}{2}} \Lambda_{1,m} \left(\sum_{i} \Lambda_{1,i} \right)^{-\frac{1}{2}}.
    \end{aligned}
\]  
\begin{claim}\label{Claim:N_to_N_0_translation}
    Define a channel $\mathcal{N}_0^{A \rightarrow C} (\sigma^A) := \mathcal{N}^{AB \rightarrow C} (\sigma^A \otimes \varphi_2^B)$.    
    Then, $\Pi_1^{E_CC}$ defined above, is an optimal tester for $(I, \epsilon, \varphi_1^{AE_C}, \mathcal{N}_0)$.
\end{claim}
\begin{proof}
Note that $\mathcal{N}$ does not act on $F_C$.
    The proof then follows directly from definition~\ref{def:optimal_projector}, and equations \eqref{eq:pi_1_ecc_a},\eqref{eq:pi_1_ecc_b}.
\end{proof}
\begin{claim} \label{claim:Alice_decoding_Correctly}
    \[\begin{aligned}
    \Tr \left[\left(\I^{E_{A_1}\ldots E_{A_M}}\otimes \Omega_{1,m}\right)\mathcal{N}^{AB\to C}\left(\varphi_1^{E_{C_m}A}\otimes \varphi_2^{B}\right)\right.  \left.\bigotimes\limits_{i\neq m} {\varphi_1}^{E_{A_i}E_{C_i}}\otimes \textsc{junk}^{E_{A_m}}\right]  \geq 1-16 \epsilon.\end{aligned}\]
\end{claim}
\begin{proof} It follows from Claim~\ref{Claim:N_to_N_0_translation} and Fact~\ref{fact:Anshu_achievability}, that,
$\Tr\left( \Omega_{1,m} \Theta_{1,m} \right) \geq 1- 16 \eps$, where \[
\begin{aligned}
\Theta_{1,m} &= \bigotimes\limits_{i \neq m}\varphi_1^{E_{C_i}} \otimes \mathcal{N}_0^{A \rightarrow C} \left( \varphi_1^{A E_{C_m}}\right)\\ &= \bigotimes\limits_{i \neq m}\varphi_1^{E_{C_i}} \otimes \mathcal{N}^{A B \rightarrow C} \left( \varphi_1^{A E_{C_m}} \otimes \varphi_2^B\right).\end{aligned}\]
The second inequality follows from the definition of $\mathcal{N}_0$.
    Now,
       \begin{align*}
           & \Tr\left[\left(\I^{E_{A_1}\ldots E_{A_M}}\otimes \Omega_{1,m}\right)\mathcal{N}^{AB\to C}\left(\varphi_1^{E_{C_m}A}\otimes \varphi_2^{B}\right)\right.\\ &\left.\bigotimes\limits_{i\neq m}\ketbra{\varphi_1}^{E_{A_i}E_{C_i}}\otimes \textsc{junk}^{E_{A_m}}\right] \\
           = & \Tr \left[ \Omega_{1,m} \bigotimes\limits_{i \neq m}\varphi_1^{E_{C_i}} \otimes \mathcal{N}^{A B \rightarrow C} \left( \varphi_1^{A E_{C_m}} \otimes \varphi_2^B\right) \right]\\
           = & \Tr\left(\Omega_{1,m} \Theta_{1,m}\right) 
           \\  \geq  &\ 1-16\eps. \qedhere
       \end{align*}  
\end{proof}
\begin{lemma} \label{lemma:after_alice_decoding}
    Let $\widehat{\Theta}_1$ be the \emph{(}global\emph{)} state after step ~\ref{item:alice_decoding}\emph{(}in Protocol~\ref{fig:multiparty_decoding}\emph{)} and \[\begin{aligned}
\Theta_{\textsc{ideal}}\coloneqq &\ketbra{m,n}^{MN}\otimes \mathcal{N}^{AB\to C}\left(\varphi_1^{E_{C_m}A}\otimes \varphi_2^{F_{C_n}B}\right)\otimes \ketbra{m}^{\widehat{M}} \bigotimes\limits_{i\neq m}{\varphi_1}^{E_{A_i}E_{C_i}}\bigotimes\limits_{j\neq n} {\varphi_2}^{ F_{B_j}F_{C_j}}
.\end{aligned}\]
    Then,
 
    \begin{enumerate}
        \item $\mathsf{Pr}\left(\widehat{M} =m \vert M=m\right)_{\widehat{\Theta}_1} \geq 1-16 \epsilon.$
        \item $
        \Vert \widehat{\Theta}_1 -  \Theta_{\textsc{ideal}} \Vert_1 \leq
      12 \sqrt{\epsilon}.$
    \end{enumerate}
\end{lemma}
\begin{proof}
Suppose Alice wants to send a message $m$ and Bob wants to send $n$. 
The global joint state just after the encoding can be described as follows:
\[
\begin{aligned}
&\ketbra{m,n}^{MN}\otimes \left(\varphi_1^{E_{C_m}A}\otimes \varphi_2^{F_{C_n}B}\right) \bigotimes\limits_{i\neq m}\ketbra{\varphi_1}^{E_{A_i}E_{C_i}} \bigotimes\limits_{j\neq n} \ketbra{\varphi_2}^{ F_{B_j}F_{C_j}} .
\end{aligned}
\]
Recall that while decoding Alice's message, Charlie disregards any of Bob's register (other than $B$ which is taken as input to the channel).
It follows from Claim~\ref{claim:Alice_decoding_Correctly} that:

\[
\begin{aligned}
\Tr&\Bigg[\ketbra{m,n}^{MN}\bigotimes  \left.\left(\Omega_{1,m}\circ \mathcal{N}^{AB\to C}\right) \left(\varphi_1^{E_{C_m}A}\otimes \varphi_2^{F_{C_n}B}\right) \bigotimes\limits_{i\neq m}\right. \left. \ketbra{\varphi_1}^{E_{A_i}E_{C_i}}\bigotimes\limits_{j\neq n} \ketbra{\varphi_2}^{ F_{B_j}F_{C_j}}\right] \\
\geq~& 1-16\eps
\end{aligned}
\]
And hence, \[\mathsf{\Pr}\left( \widehat{M} = m ~ |~ M = m \right)_{\widehat{\Theta}_1} \geq 1- 16 \epsilon.\]

Then, the Gentle Measurement Lemma (\cref{fact:gentle_measurement}) implies that the post measurement state $\widehat{\Theta}_1$ is close to the ideal state 
\[\begin{aligned}
\Theta_{\textsc{ideal}}\coloneqq &\ketbra{m,n}^{MN}\otimes \mathcal{N}^{AB\to C}\left(\varphi_1^{E_{C_m}A}\otimes \varphi_2^{F_{C_n}B}\right) \otimes \ketbra{m}^{\widehat{M}} \bigotimes\limits_{i\neq m}{\varphi_1}^{E_{A_i}E_{C_i}}\bigotimes\limits_{j\neq n} {\varphi_2}^{ F_{B_j}F_{C_j}}
\end{aligned}
\]
    in the $1$-norm by $3\sqrt{16 \eps} = 12 \sqrt{\epsilon}$. This concludes the proof.
\end{proof}

\subsection{Decoding Bob}
\subsubsection*{Defining the POVM $\lbrace\Omega_{2,n}\rbrace_n$: }

Consider \[
    I_H^{\eps}(F_C:E_CC)_{\I^{E_C F_C}\otimes\mathcal{N}^{AB \rightarrow C}\left(\varphi_1^{AE_C}\otimes \varphi_2^{BF_C}\right)} \hspace{2cm}\]

    Let $\Pi_2^{F_CE_CC}$ denote an optimal measurement for the above quantity.
    That is,
    
\begin{align} \label{eq:pi_2_ecc_a}
    &2^{- I^{\epsilon}_{H}(F_C:E_CC)_{\mathcal{N}\left(\varphi_1^{AE_C}\otimes \varphi_2^{BF_C}\right)}} =  \Tr \left[\Pi_2^{F_CE_CC }~ \left(\mathcal{N}\left( \varphi_1^{AE_C} \otimes \varphi_2^B \right) \otimes \varphi_2^{F_C} \right)\right] \\   
     &\Tr\left[\Pi_2^{F_CE_CC}\left(\mathcal{N}(\varphi_1^{AE_C} \otimes \varphi_2^{BF_C})\right)\right]\geq 1-\eps. \label{eq:pi_2_ecc_b}
\end{align}
Let 
    \begin{align*}
        \Lambda_{2,n}&= \I^{F_{C_1}} \otimes \cdots \otimes \I^{F_{C_{n-1}}} \otimes \Pi_2^{CE_CF_{C_n}} \otimes \I^{F_{C_{n+1}}} \otimes  \cdots \otimes \I^{F_{C_{2^{R_2}}}} ,  \\
        \intertext{and}
    \Omega_{2,n} &= \left(\sum_{j} \Lambda_{2,j} \right)^{-\frac{1}{2}} \Lambda_{2,n} \left(\sum_{i} \Lambda_{2,j} \right)^{-\frac{1}{2}}.
    \end{align*}
    \begin{claim}\label{Claim:N_to_N_1_translation}
    Define a channel $\mathcal{N}_1^{B \rightarrow C E_{C_m}} (\sigma^B) := \mathcal{N}^{AB \rightarrow C} (\sigma^B \otimes \varphi_1^{AE_{C_m}})$.    
    Then, $\Pi_2^{F_CE_CC}$ defined above, is an optimal tester for $(I, \epsilon, \varphi_2^{B F_{C}}, \mathcal{N}_1)$.
\end{claim}
\begin{proof}
    The proof follows directly from definition~\ref{def:optimal_projector}, equation~\eqref{eq:pi_2_ecc_a} and \eqref{eq:pi_2_ecc_b} and noting that $\mathcal{N}$ does not act on $F_C$.
\end{proof}
    \begin{claim} \label{claim:Bob_decoding_correctly}
    Let $\widehat{\Theta}_2$ be the state at the end of protocol~\ref{fig:multiparty_decoding}. Then it holds that
    \[
    \norm{\widehat{\Theta}_2^{\widehat{M}\widehat{N}}-\ketbra{m}^{\widehat{M}}\otimes \ketbra{n}^{\widehat{N}}}_1\leq \ 28 \sqrt{\eps} .
    \]
    \end{claim}
    
    \begin{proof}
    Let $\Theta_{2,\textsc{IDEAL}}$ be the post measurement state obtained by applying the POVM $\lbrace \Omega_{2,n} \rbrace_{n}$ to the state $\Theta_{\textsc{IDEAL}}$.
    By using \cref{fact:Anshu_achievability} on the channel $\mathcal{N}_1$ and the ideal state $\Theta_{\textsc{ideal}}$, we see that,
    \begin{equation} \label{eq:final_ideal_error}
         \Pr\left[\widehat{N}\neq n~|~N=n, M=m\right]_{\Theta_{2,\textsc{ideal}}} \leq 16\eps.
    \end{equation}
 Now,
\begin{align*}
    &\Pr\left[\widehat{N}\neq n~|~N=n, M=m\right]_{\widehat{\Theta}_2}\\&  \leq \Pr\left[\widehat{N}\neq n~|~N=n, M=m\right]_{\Theta_{2,\textsc{ideal}}} + \Vert \Theta_{2,\textsc{ideal}} - \widehat{\Theta}_2\Vert_1 \\
    & \leq 16 \epsilon + \Vert \Theta_{\textsc{ideal}} - \widehat{\Theta}_{1}\Vert_1 \\
    & \leq 16 \epsilon + 12 \sqrt{\epsilon} \\ & \leq 28 \sqrt{\epsilon}. 
\end{align*}
The second inequality follows from eq~\eqref{eq:final_ideal_error} and data processing.
The third inequality follows from Lemma~\ref{lemma:after_alice_decoding}.
   This concludes the proof.
    \end{proof}
    
\subsection{A More General Situation}
We note that the decoding procedure outlined in \cref{fig:multiparty_decoding} also works in a more general case, which we describe below:\\
Consider a pure state  $\ket{\varphi}^{E_CF_CAB}$ and the purifications $\ket{\varphi_1}^{E_CA}$ and $\ket{\varphi_2}^{F_CB}$ of the states $\varphi^{E_C}$ and $\varphi^{F_C}$ respectively. Consider the following situation:
\begin{enumerate}
\captionof{table}{General Decoding}\label{table:generaldecoding}
    \item Fix an index $(m,n)$.
    \item Let Alice share $2^{R_1}$ the states
    \[
    \bigotimes \limits_{i \neq m} \ket{\varphi_1}^{E_{C_i}E_{A_i}}
    \]
    with Charlie, where as before, the systems $E_{A_i}$ belong to Alice and $E_{C_i}$ belong to Charlie. Note also that $E_{A_i}\equiv A$.
    \item Similarly, let Bob share $2^{R_2}$ the  states
    \[
    \bigotimes \limits_{j \neq n} \ket{\varphi_1}^{F_{C_j}F_{B_j}}
    \]
    with Charlie, where as before, the systems $F_{B_j}$ belong to Alice and $F_{C_j}$ belong to Charlie. Note also that $F_{B_j}\equiv B$.
    \item For $i=m$ and $j=n$, let Alice, Bob, and Charlie share the tripartite state
    \[
    \mathcal{N}^{AB\to C}\left(\ketbra{\varphi}^{E_{C_m}F_{C_n}AB}\right)
    \]
    \item Then, to decode the indices $m$ and $n$, Charlie runs the protocol outlined in \cref{fig:multiparty_decoding}, with a suitable setting of decoders $\brak{\Omega_{1,m}}$ and $\brak{\Omega_{2,n}}$.
\end{enumerate}

Then, the following claim can be proved along similar lines to the proof of \cref{claim:final_decoding}:

\begin{claim}\label{claim:generaldecoding}
For any \[
    \begin{aligned}
    R_1  &\leq  I_{H}^\epsilon (E_C: C)_{\mathcal{N}\left(\varphi^{ABE_CF_C}\right)} - 2 \log\left( \frac{1}{\epsilon}\right), \\
  R_2  &\leq  I_{H}^\epsilon (F_C: C E_C)_{\mathcal{N}\left(\varphi^{ABE_CF_C}\right)} - 2 \log\left( \frac{1}{\epsilon}\right)
    \end{aligned}
    \] 
    there exist choices for the decoders $\brak{\Omega_{1,m}}$ and $\brak{\Omega_{2,n}}$ in Procedure \cref{table:generaldecoding} such that the following holds:
 \[\mathsf{Pr}\left[\left(\widehat{M},\widehat{N}\right) \neq \left(m,n\right) \vert \left(M, N\right)=\left(m, n\right) \right] \leq 28  \sqrt{\epsilon}.\] 
\end{claim}

\section{Useful Lemmas}
\begin{claim}\label{claim:highprobability}
If Alice measures the $Q$ registers of of the state $\ket{\varphi}^{\otimes n}$, where $n= \frac{1}{\delta}\cdot 2^{\imax}$, she obtains a string with at least one $0$ with probability at least $1-e^{-1/\delta}$.
\end{claim}
\begin{proof}
The probability that Alice gets all $1$'s is 
\[
\begin{aligned}
\left(1-\frac{1}{2^{\imax}}\right)^{n} &\leq e^{-n/2^{\imax}} = e^{-1/\delta}. 
\end{aligned} \qedhere
\]
\end{proof}
\begin{claim}\label{claim:extendedreduction}
Consider a protocol $\mathcal{P}=\left(M,\mathcal{N}, \mathcal{E},\mathcal{D},\ket{\varphi}^{E_AE_B}\right)$ such that the encoder $\mathcal{E}$ can toss its own private coins and abort the protocol with probability $p<1$. Suppose we are promised that, whenever the protocol does not abort, it creates the state $\rho^A$, averaged over all other systems, on the input to the channel. We are also promised that whenever the protocol does not abort, the decoder makes an error at most $\eps$ while decoding. Then, it holds that, if $M$ is the number of messages that $\mathcal{P}$ can send through the channel, then
\[
\log M \leq I_H^{\eps}\left(B:B^{\prime}\right)_{\mathcal{N}(\tau^{AB^{\prime}})}
\]
where $\ket{\tau}^{AB^{\prime}}$ is an arbitrary purification of $\rho^A$.
\end{claim}
\begin{proof}
 Define the event $E$ to be the set of those coin tosses of the encoder $\mathcal{E}$ when the protocol $\mathcal{P}$ does not abort. Let  $\mathcal{P}\vert_{E}$be the execution of the protocol $\mathcal{P}$ conditioned on the coin tosses in $E$, i.e., the encoder samples its private coins from a distribution that is supported only on the set $E$. By the promise given in the statement of the claim, $\mathcal{P}\vert_{E}$ creates the state $\rho^A$ on the input to the channel. Therefore it holds that
 \[
 \mathcal{P}\vert_{E}\in \mathcal{S}^{\rho^A}(\mathcal{N},\eps).
 \]
This implies that the total number of bits that the protocol $\mathcal{P}\vert_{E}$ can send, with the probability of error at most $\eps$ is at most $I_H^{\eps}(B:B^{\prime})_{\mathcal{N}(\tau^{AB^{\prime}})}$, by \cref{corol:usableconverse}. Note, however, that the protocol $\mathcal{P}$ does not send any bits when the coin tosses of the encoder land outside of $E$. Therefore, the total number of bits that the protocol $\mathcal{P}$ can send is, at most
\[\log M\leq I_H^{\eps}(B:B^{\prime})_{\mathcal{N}(\tau^{AB^{\prime}})}.\]
 This concludes the proof.
\end{proof}
\begin{claim}
If Alice measures the registers $Q_1 Q_2 \ldots Q_n$ in random order, where $n= \frac{1}{\delta}2^{\imax}$, then, conditioned on getting at least one $0$ outcome, it holds that
\[
\Pr\left[\textup{ success at } i~|~\textup{ success }\right] =\frac{1}{n}.
\]
\end{claim}
\begin{proof}\label{claim:distribution}
Fix a permutation $\sigma$ of the set $[n]$. Suppose
\[
\sigma(i)= j
\]
i.e., the $i$-th index is measured at time $j\in [n]$. Then,
\[
\Pr\left[\textup{ success at time } j~|~\sigma\right] = \left(1-\frac{1}{2^{\imax}}\right)^{j-1}\cdot \frac{1}{2^{\imax}}
\]
Then,
\begin{align*}
    \Pr\left[\textup{ success at } i\right] &= \sum\limits_j\sum\limits_{\sigma~|~\sigma(i)=j} \Pr\left[\textup{ success at time } j~|~\sigma\right]\cdot \Pr\left[\sigma\textup{ s.t. } \sigma(i)=j\right] \\
    &= \sum\limits_{j} \left(1-\frac{1}{2^{\imax}}\right)^{j-1}\cdot \frac{1}{2^{\imax}} \cdot \frac{1}{n} \\
    &= \left(1-\left(1-\frac{1}{2^{\imax}}\right)^{n}\right)\cdot \frac{1}{n}
\end{align*}
Also, note that
\[
\Pr\left[\textup{ success }\right] = \left(1-\left(1-\frac{1}{2^{\imax}}\right)^{n}\right)
\]
Therefore,
\[
\Pr\left[\textup{ success at } i~|~\textup{ success }\right]= \frac{1}{n}.
\]
This concludes the proof.
\end{proof}
\begin{claim}\label{claim:IHgeneral}
Given the states $\rho^A, \sigma^A$ and a maximally correlated state 
\[
\varrho^{X_1X_2}\coloneqq\frac{1}{K}\sum\limits_{x}\ketbra{x}^{X_1}\otimes \ketbra{x}^{X_2},
\]
it holds that
\[
D_H^{\eps}(\rho^A\otimes \varrho^{X_1X_2}~||~\sigma^A\otimes \pi^{X_1}\otimes \pi^{X_2}) \leq D_H^{\sqrt{\eps}}(\rho^A~||~\sigma^A)+\log K-\log (1-\sqrt{\eps}).
\]
\end{claim}
\begin{proof}
Let $\Pi^{AX_1X_2}_{\textsc{opt}}$ be {an} optimal {tester for } $d:=D_H^{\eps}(\rho^A\otimes \varrho^{X_1X_2}~||~\sigma^A\otimes \pi^{X_1}\otimes \pi^{X_2})$.

\vspace{1mm}

That is,
\begin{align}
    \Tr\left[ \Pi^{AX_1X_2}_{\textsc{opt}} \left(\rho^A\otimes \varrho^{X_1X_2} \right) \right] &\geq 1 -\epsilon \\
    \Tr\left[ \Pi^{AX_1X_2}_{\textsc{opt}} \left(\sigma^A\otimes \pi^{X_1}\otimes \pi^{X_2} \right) \right] &\leq 2^{-d}
\end{align}

Define
\[
\Pi^{A}_{x_1,x_2}\coloneqq \left(\I^A\otimes \bra{x_1,x_2}^{X_1X_2}\right)\Pi_{\textsc{opt}}^{AX_1X_2}\left(\I^A\otimes \ket{x_1,x_2}^{X_1X_2}\right)  
\]
It is then easy to see that
\[
\begin{aligned}
1- \epsilon & \leq  \Tr\left[\Pi_{\textsc{opt}}^{AX_1X_2} ~ \rho^A\otimes \varrho^{X_1X_2}\right]\\
& = \Tr\left[\Pi_{\textsc{opt}}^{AX_1X_2} ~ \rho^A\otimes \sum\limits_{x}\frac{1}{K}\ketbra{x,x}^{X_1X_2}\right] \\
& = \sum_x \frac{1}{K} \Tr\left[\Pi_{\textsc{opt}}^{AX_1X_2} ~ \rho^A\otimes \ketbra{x,x}^{X_1X_2}\right] \\
&= \sum_x \frac{1}{K} \Tr\left[\left(\bra{x,x}^{X_1X_2} ~  \Pi_{\textsc{opt}}^{AX_1X_2} \ \ket{x,x}^{X_1X_2}\right)^A\rho^A\right] \\
&= \sum_x \frac{1}{K} \Tr\left[  \Pi_{x,x}^A\rho^A\right] \\
\end{aligned}
\]
We define a set $\good_X$ as follows:
\[ \good_X = \lbrace x \ : \ \Tr\left[\Pi^A_{x,x}~\rho^A\right] \geq 1-\sqrt{\eps}\rbrace.\]
A standard Markov argument then gives that
\[
\abs{\good_X}\geq (1-\sqrt{\eps})  K.
\]
\noindent Again, note that

\begin{align*}
\Tr\left[\Pi^{AX_1X_2}_{\textsc{opt}}\left(\sigma^A\otimes \pi^{X_1}\otimes \pi^{X_2}\right)\right] =  &\sum\limits_{x_1,x_2}  \frac{1}{K^2} \Tr\left[\Pi^{AX_1X_2}_{\textsc{opt}}\left(\sigma^A\otimes \ketbra{x_1}^{X_1} \otimes \ketbra{x_2}^{X_1} \right)\right] \\
 = & \sum\limits_{x_1,x_2} \frac{1}{K^2} \Tr_{A} \left[ \Tr_{X_1 X_2} \Pi^{AX_1X_2}_{\textsc{opt}}\left(\sigma^A\otimes \ketbra{x_1}^{X_1} \otimes \ketbra{x_2}^{X_2} \right) \right] \\
 = & \sum_{x_1, x_2} \frac{1}{K^2} \Tr_{A} \left[ \bra{x_1 x_2}^{X_1 X_2} ~  \Pi^{AX_1X_2}_{\textsc{opt}} ~  \ket{x_1 x_2}^{X_1 X_2} \sigma^{A} \right]
\\  = & \sum_{x_1, x_2} \frac{1}{K^2} \Tr_{A} \left[ \Pi^{A}_{x_1 x_2} ~   \sigma^{A} \right]
\\
\geq  &~ \sum_{x\in \good_X}\frac{1}{K^2}\Tr\left[\Pi^A_{x,x}~\sigma^A\right] \\
\geq &~ \sum_{x\in \good_X} \frac{1}{K^2} 2^{-D_H^{\sqrt{\eps}}(\rho^A~||~\sigma^A)} \\
\geq &~ \frac{1-\sqrt{\eps}}{K}\cdot 2^{-D_H^{\sqrt{\eps}}(\rho^A~||~\sigma^A)} \qedhere
\end{align*}  
\end{proof}
 \end{document}